\documentclass[review]{elsarticle}
\usepackage{a4wide}
\usepackage[margin=2.7cm]{geometry}
\usepackage{amsmath}
\usepackage{amsfonts}
\usepackage{amssymb}
\usepackage{amsthm}
\usepackage{caption}
\usepackage{bbm}
\usepackage[english]{babel}
\usepackage{graphicx}
\usepackage{epstopdf}
\usepackage{amsbsy}
\usepackage{hyperref}
\newtheorem{thm}{Theorem}
\newtheorem{prp}{Proposition}
\usepackage{framed}

\journal{arXiv}
\begin{document}

\begin{frontmatter}

\title{A Data Propagation Model for Wireless Gossiping}

\author{T.M.M. Meyfroyt, S.C. Borst, O.J. Boxma}
\address{Eindhoven University of Technology\\
P.O. Box 513, 5600 MB Eindhoven, The Netherlands\\
\{t.m.m.meyfroyt, s.c.borst, o.j.boxma\}@tue.nl}

\author{D. Denteneer}
\address{Philips Research\\
HTC 34, 5656 AE Eindhoven, The Netherlands\\
dee.denteneer@philips.com}
\begin{abstract}
Wireless sensor networks require communication protocols for efficiently propagating data in a distributed fashion. The Trickle algorithm is a popular protocol serving as the basis for many of the current standard communication protocols.  In this paper we develop a mathematical model describing how Trickle propagates new data across a network consisting of nodes placed on a line. The model is analyzed and asymptotic results on the hop count and end-to-end delay distributions in terms of the Trickle parameters and network density are given. Additionally, we show that by only a small extension of the Trickle algorithm the expected end-to-end delay can be greatly decreased. Lastly, we demonstrate how one can derive the exact hop count and end-to-end delay distributions for small network sizes.
\end{abstract}


\end{frontmatter}

\section{Introduction}
In recent years wireless sensor networks have been quickly growing in popularity. In these networks  inexpensive autonomous sensor units, called nodes, gather data, which they exchange with each other through wireless transmissions. Wireless sensor networks have various applications, for example in health care, area- and industrial monitoring \cite{trick1}. 

Typically, these networks are communication, computation, memory and energy constrained and therefore require appropriate distributed communication protocols. The requirements for a good communication protocol are threefold. First, it should be able to quickly disseminate and collect data within the network. Second, this should be done as efficiently as possible, meaning that the number of transmissions in the network should be as low as possible. Third, communication should be reliable; new data should eventually be received by all the nodes in the network. Several communication protocols have been proposed in recent years for this purpose, see for example \cite{trick13, trick17, trick3, trick11, trick12}. 

In \cite{trick3} the Trickle algorithm has been proposed in order to effectively and efficiently distribute and maintain information in such networks. Trickle relies on a ``polite gossip" policy to quickly propagate updates, while minimizing the number of redundant transmissions. Because of its broad applicability, Trickle has been documented in its own IETF RFC 6206 \cite{trick4}. Moreover, it has been standardized as part of the IPv6 Routing Protocol for Low power and lossy networks \cite{trick2} and the Multicast Protocol for Low power and lossy networks \cite{trick9}.

Because the Trickle algorithm has become a standard and is being widely used, it is crucial to understand how its parameters affect QoS measures like energy usage and end-to-end delay. However, not much work has yet been done in this regard. Most of the papers that evaluate Trickle try to give guidelines on how to set the Trickle parameters using simulations \cite{sims, DIS, trick5, trick3}.

The goal of this paper is to develop and analyze an analytical model describing how Trickle disseminates updates throughout a network. Our results serve as a first step towards understanding how Trickle's parameters influence its performance; however, the main focus of this paper will be on the mathematical analysis. Additionally, our models are relevant for the analysis of protocols that build upon Trickle, such as Deluge \cite{trick15} and Melete \cite{trick12}.

\subsection{Key contributions of the paper}
As key contributions of this paper, we thoroughly analyze a Markov renewal process which models a Trickle propagation event in networks of nodes arranged on a line. We show how the hop count and end-to-end delay distributions depend on the Trickle parameters and network density, as the size of the network grows large. These insights help us to better understand the impact of Trickle's parameters on its performance and how to tune these parameters. Furthermore, we show that by a simple extension of the Trickle algorithm, the expected end-to-end delay can be significantly decreased. Finally, we show how to calculate the generating functions of the hop count and end-to-end delay distributions for any network size.

\subsection{Related work}

Some analytical results concerning the message count of the Trickle algorithm are provided in \cite{trick6, trick5, meyfroyt}. First, in \cite{trick5} qualitative results are provided on the scalability of the algorithm. Specifically, it is conjectured that in single-hop networks the number of transmissions per time interval is bounded regardless of the network size, if a listen-only period is used. When no listen-only period is used, the message count scales as $\mathcal{O}(\sqrt{n})$, where $n$ is the size of the network. These claims are proven in \cite{meyfroyt} and tight upper bounds and accompanying growth factors for the message count are provided. Additionally, distributions of times between consecutive transmissions for large single-hop networks are derived and the concept of a listen-only period is generalized. Moreover, approximations for the message count in multi-hop networks are provided. Lastly, in \cite{trick6}, a different, slightly more accurate, approximation is given for the message count in multi-hop networks with specific parameter settings. 

Analytical results on the speed at which the Trickle algorithm can propagate new data throughout a network are even fewer. In \cite{trick7} the authors provide a method for deriving the Laplace transform of the distribution function of the end-to-end delay for any network topology. However, the method is computationally involved, limiting its practical use, and does not provide much insight. In this paper we provide more easily computable expressions for the Laplace transforms of the hop count and end-to-end delay in a line network.

\subsection{Organization of the paper}
The remainder of this paper is organized as follows. In Section \ref{Trickle} we give a detailed description of the Trickle algorithm and propose a simple extension.   We then analyze how fast the algorithm can propagate updates across a network consisting of nodes placed on a line in Section \ref{model}. Additionally, we derive the limiting distributions for the hop count and end-to-end delay as the size of the network grows large. These results are compared with simulations and we show that with our modified algorithm the expected end-to-end delay can be significantly decreased.  This is followed in Section \ref{gen} by a derivation of the generating functions of the hop count and end-to-end delay for any network size. Finally, we present our conclusions in Section \ref{conclusion}.

\section{The Trickle Algorithm}\label{Trickle}
We now provide a detailed description of the original Trickle algorithm \cite{trick3} before introducing a  small extension, which helps improve the algorithm's performance. The Trickle algorithm has two main goals. First, whenever a new update enters the network, it must be propagated quickly. Secondly, when there are no updates, communication overhead has to be kept to a minimum. 

The Trickle algorithm achieves this by using a ``polite gossip" policy. Nodes divide time into intervals of varying length. During each interval a node will broadcast its current information, if it has not heard other nodes transmit the same information during that interval, in order to check if its information is up to date. If it has recently heard another node transmit the same information it currently has, it will stay quiet, assuming there is no new information to be received. Additionally, it will increase the length of its intervals, decreasing its broadcasting rate. Whenever a node receives an update or hears old information, it will reduce its interval size, increasing its broadcasting rate, in order to quickly resolve the inconsistency. This way inconsistencies are detected and resolved fast, while keeping the number of transmissions low.

\subsection{Algorithm Description}
The algorithm has three parameters:
 \begin{itemize}
 \item A threshold value $k$, called the redundancy constant.
 \item The maximum interval size $\tau_h$.
 \item The minimum interval size $\tau_l$.

 \end{itemize}
Furthermore, each node in the network has its own timer and keeps track of three variables:
  \begin{itemize}
 \item The current interval size $\tau$.
 \item A counter $c$, counting the number of messages heard during an interval.
 \item A broadcasting time $t$ during the current interval.
 \end{itemize}
The behavior of each node is described by the following set of rules:
\begin{enumerate}
 \item At the start of a new interval a node resets its timer and counter $c$ and sets $t$ to a value in $[\frac{1}{2}\tau,\tau]$ at random.
 \item When a node hears a message that is consistent with the information it has, it increments $c$ by 1.
 \item When a node's timer hits time $t$, the node broadcasts its message, if $c<k$.
 \item When a node's timer hits time $\tau$, it doubles its interval size $\tau$ up to $\tau_h$ and starts a new interval.
 \item When a node hears a message that is inconsistent with its own information, then if $\tau>\tau_l$ it sets $\tau$ to $\tau_l$ and starts a new interval, otherwise it does nothing.
\end{enumerate}

\subsection{Modification to the algorithm}
Consider rule 1 of the Trickle algorithm. It states that nodes should pick their broadcasting times uniformly in $[\frac{1}{2}\tau,\tau]$ at random, leaving $\tau/2$ time units before broadcasting as a listen-only period. The reason for having this listen-only period is discussed in \cite{trick3}. The authors of \cite{trick3} argue that when no listen-only period is used, i.e. nodes always pick $t$ in $[0,\tau]$, sometimes nodes will broadcast soon after the beginning of their interval, listening for only a short time, before anyone else has a chance to speak up. If we have a perfectly synchronized network this does not give a problem, because the first $k$ transmissions will simply suppress all the other broadcasts during that interval. However, in an unsynchronized network, if a node has a short listening period, it might broadcast just before another node starts its interval and that node possibly also has a short listening period. This possibly leads to a lot of redundant messages and is referred to as the short-listen problem.

It has been shown in \cite{meyfroyt}, that such a listen-only period is indeed necessary to  resolve the short-listen problem and to ensure scalability of the Trickle algorithm. However, the authors of 
\cite{meyfroyt} also observe that introducing such a listen-only period can greatly affect propagation speed. That is, when a listen-only period of $\tau/2$ is used, newly updated nodes will always have to wait for a period of at least $\tau_l/2$, before attempting to propagate the received update. Consequently, in an $m$-hop network, the end-to-end delay is at least $m\tau_l/2$. Hence, as is also argued in \cite{meyfroyt}, on the one hand long listen-only periods reduce the number of redundant transmissions, but on the other hand short listen-only periods increase propagation speed. 

For these reasons, based on ideas from \cite{meyfroyt}, we propose to add a listen-only parameter $\eta$ and to modify rule 1 of the Trickle algorithm:
\begin{framed}
\begin{enumerate}
 \item[$1^*$.] At the start of a new interval a node resets its timer and counter $c$ and, if $\tau=\tau_l$, sets $t$ to a value in $[\eta\tau,\tau]$ at random, otherwise in $[\frac{1}{2}\tau,\tau]$ at random.
 \end{enumerate}
\end{framed}
Here one should think of $\eta$ being smaller than $1/2$ (note that $\eta=1/2$ gives the original Trickle algorithm) and, as our analysis will show, preferably $\eta=0$. Rule $1^*$ then tries to achieve the best of both worlds. When nodes have just received an update and reset $\tau$ to $\tau_l$, they are allowed to be impatient and transmit after listening for only $\eta\tau$ time units. Probably, they are at the front of the propagation wave and have neighbors that are not yet up to date. When $\tau>\tau_l$, the wave front probably has passed, and nodes should first listen to what their neighbors have to say, before deciding whether to broadcast or not.

Furthermore, note that even for $\eta=0$ this modification does not suffer from the short-listen problem. When a group of nodes is updated by a broadcast, they will all start a new interval at the same time and therefore become synchronized. Therefore, the short-listen problem will not cause additional redundant transmissions within this set of nodes. Furthermore, since at the front of the propagation wave nodes are almost synchronized, having them pick $t$ in $[\eta\tau,\tau]$ as opposed to picking $t$ in $[\tau/2,\tau]$, gives the newly updated nodes a bigger contention window to schedule broadcasts, leading to fewer channel collisions. After the update has passed, nodes will go back to having a listen-only period of half an interval.

\section{Propagation model}\label{model}
In this section we develop and analyze a model describing how fast the modified Trickle algorithm can propagate updates in a network consisting of nodes placed on a line. We will first briefly discuss the assumptions of our model, their relevance and their limitations. This is followed by an analysis of the model and lastly we validate our results through simulations.

\subsection{Model assumptions}
As mentioned, the goal of this paper is to analytically gain insight in the performance of Trickle when used to disseminate data in wireless networks. Therefore, we will focus on networks consisting of nodes placed on a line. Understanding how Trickle disseminates data across a line should already give us useful insights in the performance of Trickle in general. Additionally, this is a common network topology in many applications, for example in intelligent street lighting. Moreover, simulation experiments in \cite{thesis} show that in regular topologies, such as grids, Trickle's performance is comparable to its performance in a line network. This is because in such networks Trickle tends to disseminate data in each direction at the same speed.

Secondly, we will concentrate on the setting $k=1$, which permits a detailed mathematical analysis. Moreover, this is a commonly used setting when Trickle is used for data dissemination, as in MPL \cite{trick9}. We do note that simulations suggest that for other regular network topologies, such as grids, performance for other $k$ is qualitatively the same as for $k=1$ and increasing $k$ has very little effect on the end-to-end delay \cite{sims, thesis, trickf}. In these scenarios, increasing $k$ is generally used in order to deal with lossy transmissions. However, recent work shows that in random topologies, when Trickle is used for routing, as in RPL, performance can strongly depend on the parameter $k$ and the exact network topology \cite{DIS, trickf}, but this is beyond the scope of this work. See \cite{DIS} for a detailed simulation study on the effect of the redundancy constant $k$ on the performance of RPL in random topologies.

Lastly, we assume all the nodes are perfect receivers and transmitters, i.e. transmissions are instantaneous and there is no packet loss. This allows us to focus on the performance of Trickle without explicit consideration of any MAC-layer protocols. Future work of the authors concerns analyzing the impact of the MAC-layer on the performance of Trickle when used as a data dissemination protocol.

\subsection{Model analysis}
We first introduce the model and some notation. Assume we have $n+1$ nodes arranged on a line and each node is separated by a distance of 1 from its neighbors. We label the nodes $0\text{, }1\text{, ..., }n$ from left to right. Without loss of generality we assume $\tau_l=1$. Updates will be injected into the network at node 0. Nodes have a fixed transmission range $R$, which means that when a node sends a message, only nodes within a distance $R$ of the broadcaster will receive the message. Finally, assume initially that all the nodes have $\tau=\tau_h$, and at time 0 an update is injected at node 0, which it starts to propagate. Let us denote the time node $n$ gets updated by $T^{(n)}$, which is called the end-to-end delay, and the number of transmissions needed to reach node $n$ by $H^{(n)}$, which is called the hop count.

Observe that because node 0 gets updated at time 0, this node will broadcast the update somewhere in the interval $[\eta,1]$, updating nodes $1$ to $R$. The newly updated nodes will reset their intervals, synchronize and set $\tau=\tau_l=1$. Node 0 will double its interval length after its interval ends and will have a listen-only period of length $\tau_l$ in its next interval. As a result, one of the newly updated nodes will be the next node to transmit. When node 1 is the first node to broadcast, it will only update node $R+1$, which will be the next broadcaster. When node $R$ is the next node to broadcast, it will update nodes $R+1$ to $2R$. After this step, nodes 1 to $R$ will double their interval length and the next transmission will again be done by one of the newly updated nodes. Let us formalize this process. 

Let $U_m$ be the number of nodes that are updated by the $m$'th broadcast and let $U_0=1$. Then we can write
\begin{equation}p_{ij}=\mathbb{P}[U_{m+1}=j\text{ }|\text{ }U_{m}=i]=\left\{\begin{array}{ll}\frac{1}{i}, & \hbox{$R-i<j\leq R$,}\\\\
   0 , & \hbox{otherwise.}\end{array} \right.
\end{equation}
Hence, $\{U_i\}_{i=0}^{\infty}$ forms a Markov chain with states $\{1\text{, ..., }R\}$ and transition matrix $P=[p_{ij}]$, which allows us to analyze the hop count of a propagation event. Calculating its steady-state probability vector $\boldsymbol{\pi}$ we find
\begin{equation}\label{pi}\boldsymbol{\pi}=\frac{2}{R(R+1)}(1\text{, ..., }R).\end{equation}
Moreover, the expected number of nodes that get updated by each hop in steady state is given by
\begin{equation}\label{muU}\mu_U=\mathbb{E}[U]=\lim_{i\rightarrow \infty}\mathbb{E}[U_i]=\sum_{j=1}^R j \pi_j=\frac{1}{3}(2R+1).\end{equation}
The time between consecutive hops depends on the number of newly updated nodes $U_i$. Let $T_i$ be the time of the $i$'th transmission and $\theta_i=T_{i}-T_{i-1}$ be the inter-transmission time between the $i-1$'th and $i$'th transmission and let $T_0=0$. We then know that the time $\theta_{i+1}$ is the minimum of $U_{i}$ Trickle timers.  Then, $\theta_{i+1} \sim \eta+(1-\eta)\beta[1,U_{i}]$, since the minimum of $m$ uniform random variables follows a $\beta(1,m)$ distribution. More precisely
\begin{equation}\label{nu}\mathbb{P}[\theta_{i+1}\leq t\text{ }|\text{ }U_{i}=u]= \left\{\begin{array}{ll}0, & \hbox{$t<\eta$,}\\\\
1-\left(\frac{1-t}{1-\eta}\right)^u , & \hbox{$\eta \leq t\leq 1$,}\\\\
1 , & \hbox{otherwise.}\end{array} \right.  
\end{equation}
Hence the expected time between transitions in steady state is given by
\begin{equation}\label{muT}\mu_\theta=\mathbb{E}[\theta]=\lim_{i\rightarrow \infty}\mathbb{E}[\theta_i]=\sum_{j=1}^R\pi_j\left(\eta+\frac{1-\eta}{j+1}\right)=\eta+2(1-\eta)\frac{R+1-\sum_{j=1}^{R+1}\frac{1}{j}}{R(R+1)}.
\end{equation}
The process $(U_i,T_i)_{i=0}^{\infty}$ is called a Markov renewal process, see \cite{renewal}. We will now analyze this Markov renewal process to gain insight in the propagation speed of the Trickle algorithm.

First note that the hop count $H^{(n)}$ can be written as
\begin{equation}\label{H}H^{(n)}=\min\left\{m\text{: }\sum_{i=1}^m U_i\geq  n\right\}.\end{equation}
Also note that $H^{(n)}$ is a stopping time with respect to the Markov chain $\{U_i\}_{i=0}^{\infty}$. Furthermore, we can write the end-to-end delay $T^{(n)}$ as follows:
\begin{equation}\label{T}
T^{(n)}=T_{H^{(n)}}=\sum_{i=1}^{H^{(n)}} \theta_i.
\end{equation}
We will analyze the behavior of the Markov renewal process for large $n$. For simplicity, we shall assume stationarity of the underlying Markov chain in some of the arguments and no longer assume $U_0=1$. However, note that the asymptotic results for large $n$ also hold for the case $U_0=1$, since we have a finite-state Markov chain, which converges geometrically fast to its steady-state distribution. Assuming stationarity of the Markov chain, we have by Wald's equation 
\[\mathbb{E}\left[H^{(n)}\right]=\mathbb{E}\left[\sum_{i=1}^{H^{(n)}} U_i\right]/\mu_U.\]
Furthermore, since $n\leq\mathbb{E}\left[\sum_{i=1}^{H^{(n)}} U_i\right]\leq n+R-1$, we find
\begin{prp}\label{prp1} For $k=1$, $\eta\in[0,1]$ and $R \in \mathbb{N}^+$,
\begin{equation}\lim_{n\rightarrow \infty} \frac{\mathbb{E}\left[H^{(n)}\right]}{n}=\frac{1}{\mu_U}=\frac{3}{2R+1} .\end{equation}
\end{prp}

Hence, since $1/R$ can be seen as a measure for the density of the network, we find that the hop count decreases linearly with the network density and is independent of the choice for $\eta$, as expected. 

Now using \eqref{T} and again applying Wald's equation, we conclude
\begin{prp}\label{prp2}For $k=1$, $\eta\in[0,1]$ and $R\in \mathbb{N}^+$,
\begin{equation}
\lim_{n\rightarrow \infty} \frac{\mathbb{E}\left[T^{(n)}\right]}{n}=\frac{\mu_\theta}{\mu_U}=\frac{3}{2R+1}\left(\eta+2(1-\eta)\frac{R+1-\sum_{j=1}^{R+1}\frac{1}{j}}{R(R+1)}\right).
\end{equation}
\end{prp}

Proposition \ref{prp2} reveals the impact of a listen-only period on the end-to-end delay. First of all, we find that the expected end-to-end delay is decreasing in $\eta$. Furthermore, if $\eta>0$, the end-to-end delay decreases linearly with the density of the network, due to each hop taking at least $\eta$ time units. If $\eta=0$ the end-to-end delay decreases quadratically with network density. Hence, in dense networks the algorithm can benefit greatly from setting $\eta=0$.

In addition to a law of large numbers for the hop count and end-to-end delay, we now provide results concerning their limiting distribution. First of all, we provide results on the asymptotics of the variance of $H^{(n)}$ and $T^{(n)}$.

\begin{thm}\label{thm1}
Let $H^{(n)}$ be as defined by (\ref{H}). Then
\begin{equation}\label{sigmaH}\lim_{n\rightarrow \infty}\frac{\textnormal{Var}\left[H^{(n)}\right]}{n}=\sigma_H^2=\frac{\gamma_U^2}{\mu_U^3}=\frac{R^2+R-2}{16R^3+24R^2+12R+2},\end{equation}
where
\begin{equation}\label{gammaU}\gamma_U^2=\lim_{m\rightarrow \infty}\frac{1}{m}\textnormal{Var}\left[\sum_{i=0}^m U_i\right]=\textnormal{Var}[U_0] + 2\sum_{j=1}^\infty \textnormal{Cov}[U_0,U_{j}].\end{equation}
\end{thm}
\begin{proof}
See \ref{appA}.
\end{proof}

Similarly, we have the following asymptotic result for the variance of $T^{(n)}$.
\begin{thm}\label{thm2}
Let $T^{(n)}$ be as defined by (\ref{T}). Then
\begin{equation}\label{sigmaT}\lim_{n\rightarrow \infty}\frac{\textnormal{Var}\left[T^{(n)}\right]}{n}=\sigma_T^2=\frac{\gamma_T^2}{\mu_U^3}=\frac{\mu_\theta^2\gamma_U^2+\mu_U^2\gamma_\theta^2-2\mu_U\mu_\theta\Delta}{\mu_U^3},\end{equation}
where $\gamma_U$ is as defined in \eqref{gammaU},
\begin{equation}\label{gammaT}\gamma_T^2=\lim_{m\rightarrow \infty}\frac{1}{m}\textnormal{Var}\left[\left(\sum_{i=0}^m U_i\right)\mu_\theta-T_{m+1}\mu_U\right],\end{equation} 
\[\Delta=(1-\eta)\frac{(4R+8)\left(\sum_{j=1}^{R+1}\frac{1}{j}\right)-(R^2+9R+8)}{9R^2+9R},\]
and 
\[\gamma_\theta^2=\textnormal{Var}[\theta_1]+2\boldsymbol{\pi}MZM\boldsymbol{1}-2\mu_\theta^2,\]
with $Z=(I-P+\boldsymbol{1}\boldsymbol{\pi})^{-1}$ the fundamental matrix and $M=\left[p_{ij}\left(\eta+\frac{1-\eta}{i+1}\right)\right]$.
\end{thm}
\begin{proof}
See \ref{appB}.
\end{proof}

Let us consider Equation \eqref{sigmaT} in more detail and investigate how the variance of the end-to-end delay depends on $\eta$ and $R$.  For this we plot $\sigma_T^2$ as a function of $\eta$ for $R=5$, $R=10$ and $R=30$ in Figure \ref{fig:2}. We find that for $R=5$ the variance is minimized for  $\eta\approx0.56$. For $R=10$, the minimum is achieved at $\eta\approx0.26$ and for $R=30$ at $\eta=0$. 
\begin{figure}[!h]
\minipage{0.32\textwidth}
  \includegraphics[width=\linewidth]{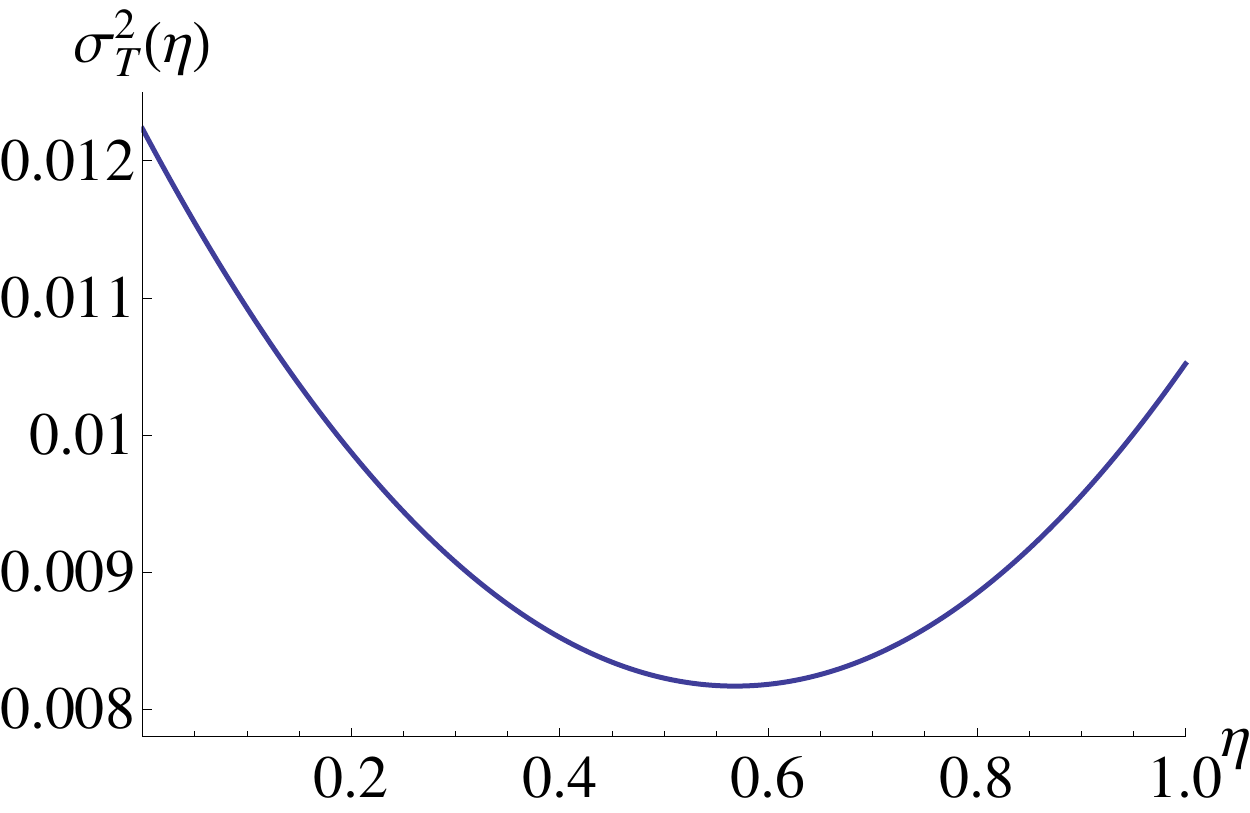}
  \caption*{\begin{small}$
R=5$
  \end{small}}
\endminipage\hfill
\minipage{0.32\textwidth}
  \includegraphics[width=\linewidth]{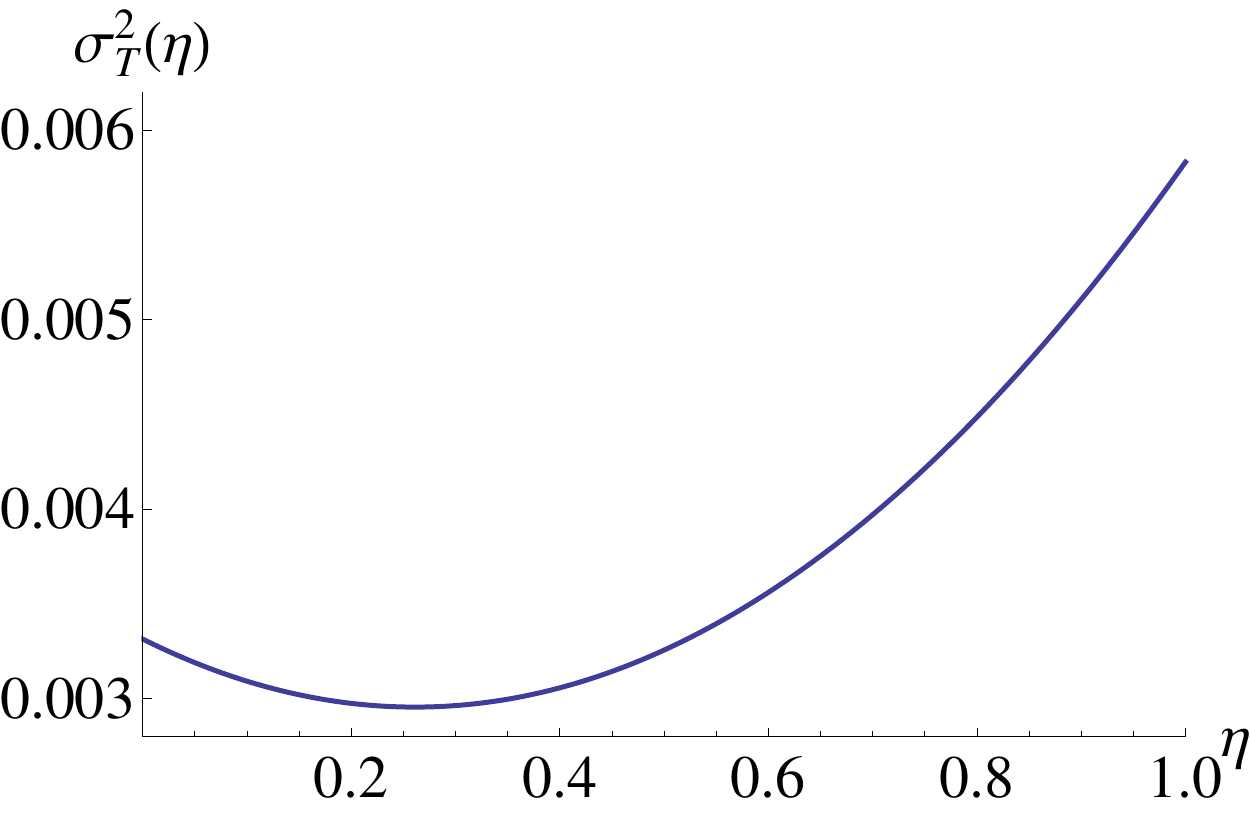}
  \caption*{\begin{small}$R=10$
  \end{small}}
\endminipage\hfill
\minipage{0.32\textwidth}%
  \includegraphics[width=\linewidth]{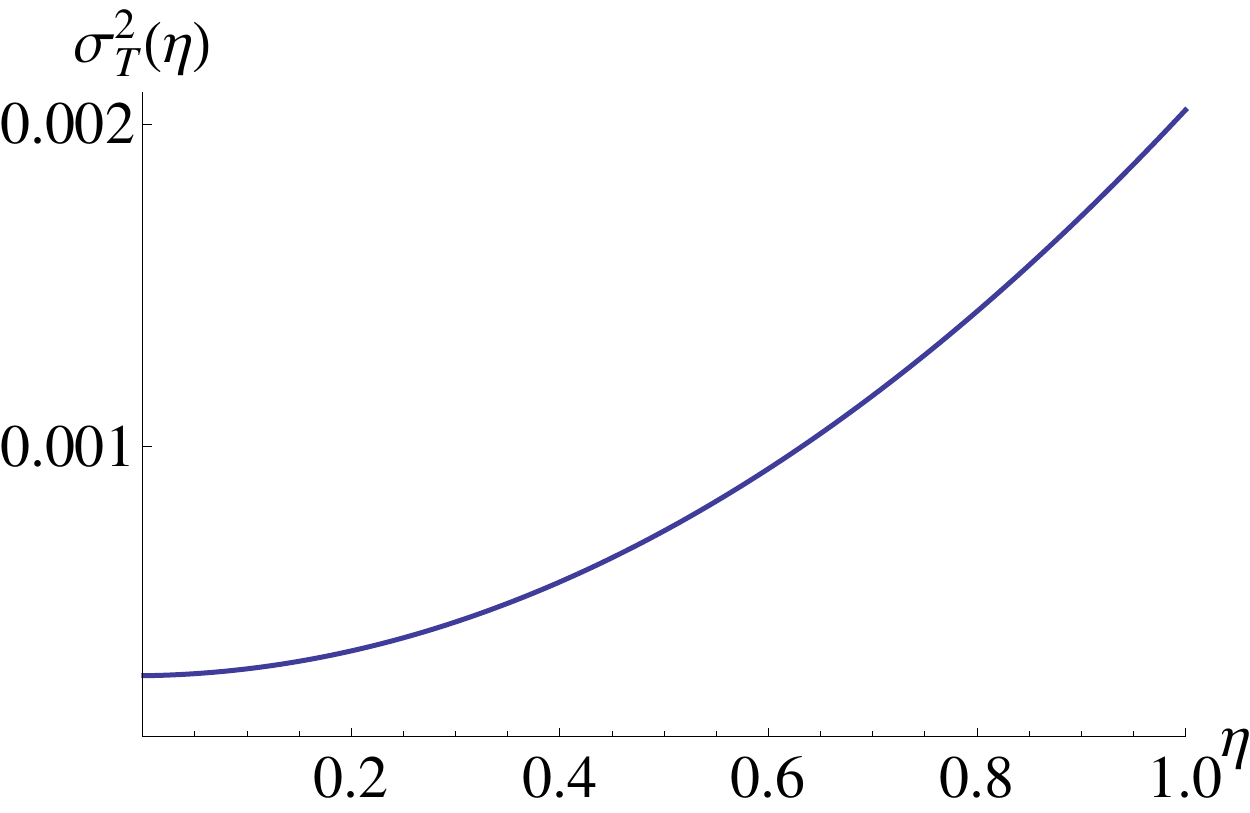}
  \caption*{\begin{small}$
R=30$
  \end{small}}
\endminipage
\caption{$\sigma_T^2$ as a function of $\eta$ for different $R$.}\label{fig:2}
\end{figure}

This can be explained as follows. For small $R$, the variance of the hop count is small compared to the variance of inter-transmission times. Consider for example the extreme case $R=1$. In this case $H^{(n)}=n$ and $\sigma_H^{(n)}=0$, while the time between hops is of some length in $[\eta\tau, \tau]$ uniformly at random. Therefore, if one wants to minimize the variance of the end-to-end delay it pays to increase $\eta$ in order to decrease the variance of inter-transmission times (at the cost of large end-to-end delays).

However, for large $R$, the hop count has high variability, while inter-transmission times have small variance. To see this, consider the other extreme case, where $R\rightarrow \infty$. The number of nodes that then get updated each hop becomes highly variable, while an inter-transmission time will always take very close to $\eta$ time units. Hence, in this case, it pays to decrease $\eta$ in order to decrease the impact of the hop count on the end-to-end delay variance. 

Therefore, for dense networks ($R$ large), setting $\eta=0$ is always preferred to $\eta>0$, since it minimizes both the expected value and variance of the propagation delay. For sparse networks ($R$ small), one can reduce the variance by setting $\eta>0$, however this greatly increases the expected delay, hence for most applications $\eta=0$ is probably more desirable.
\newline\\
Finally using results from \cite{renewal3}, we get the following results for the limiting distributions of $H^{(n)}$ and $T^{(n)}$ (see \cite{renewal3}, Theorem 1).
\begin{thm}\label{thm3}
Let $H^{(n)}$ and $T^{(n)}$ be as defined in (\ref{H}) and (\ref{T}) respectively. Then
\begin{equation}\label{hlim}\frac{H^{(n)}-\frac{1}{\mu_U} n}{\sigma_H n^{1/2}}\overset{\text{d}}{\longrightarrow}\mathcal{N}[0,1]\end{equation}
and
\begin{equation}\label{tlim}\frac{T^{(n)}-\frac{\mu_\theta}{\mu_U} n}{\sigma_T n^{1/2}}\overset{\text{d}}{\longrightarrow}\mathcal{N}[0,1],\end{equation}
where $\sigma_H$ is as defined in (\ref{sigmaH}) and $\sigma_T$ as defined in (\ref{sigmaT}).
\end{thm}
In the next section we will compare the results from Theorem \ref{thm3} with simulation results for networks with finite $n$, which will give us a better idea of the implications of Equations \eqref{hlim} and \eqref{tlim}.

\subsection{Simulation results}
We now look at some simulation results on the hop count and end-to-end delay and compare them with the asymptotic results from Theorem \ref{thm3}. We consider a sparse network with $R=5$ and $n=250$, and a dense network with $R=30$ and $n=1500$ and in both cases vary $\eta$. Note, that $n/R=30$ for both scenarios, which allows us to make a fair comparison between the two scenarios. For each scenario we run $10^5$ simulations.

\begin{figure}[!h]
\minipage{0.32\textwidth}
  \includegraphics[width=\linewidth]{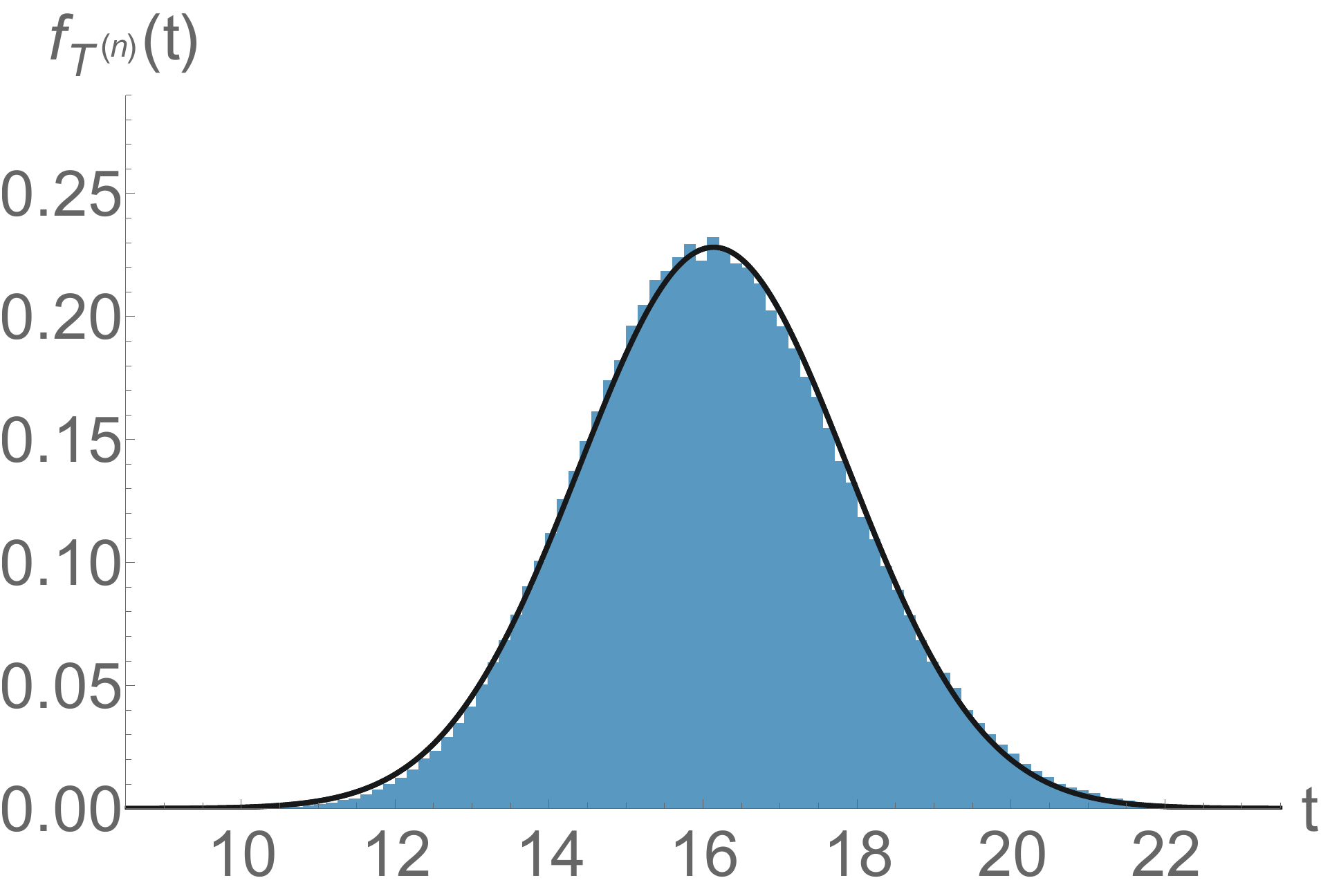}
  \caption*{\begin{small}$
  \eta=0$
  \end{small}}
\endminipage\hfill
\minipage{0.32\textwidth}
  \includegraphics[width=\linewidth]{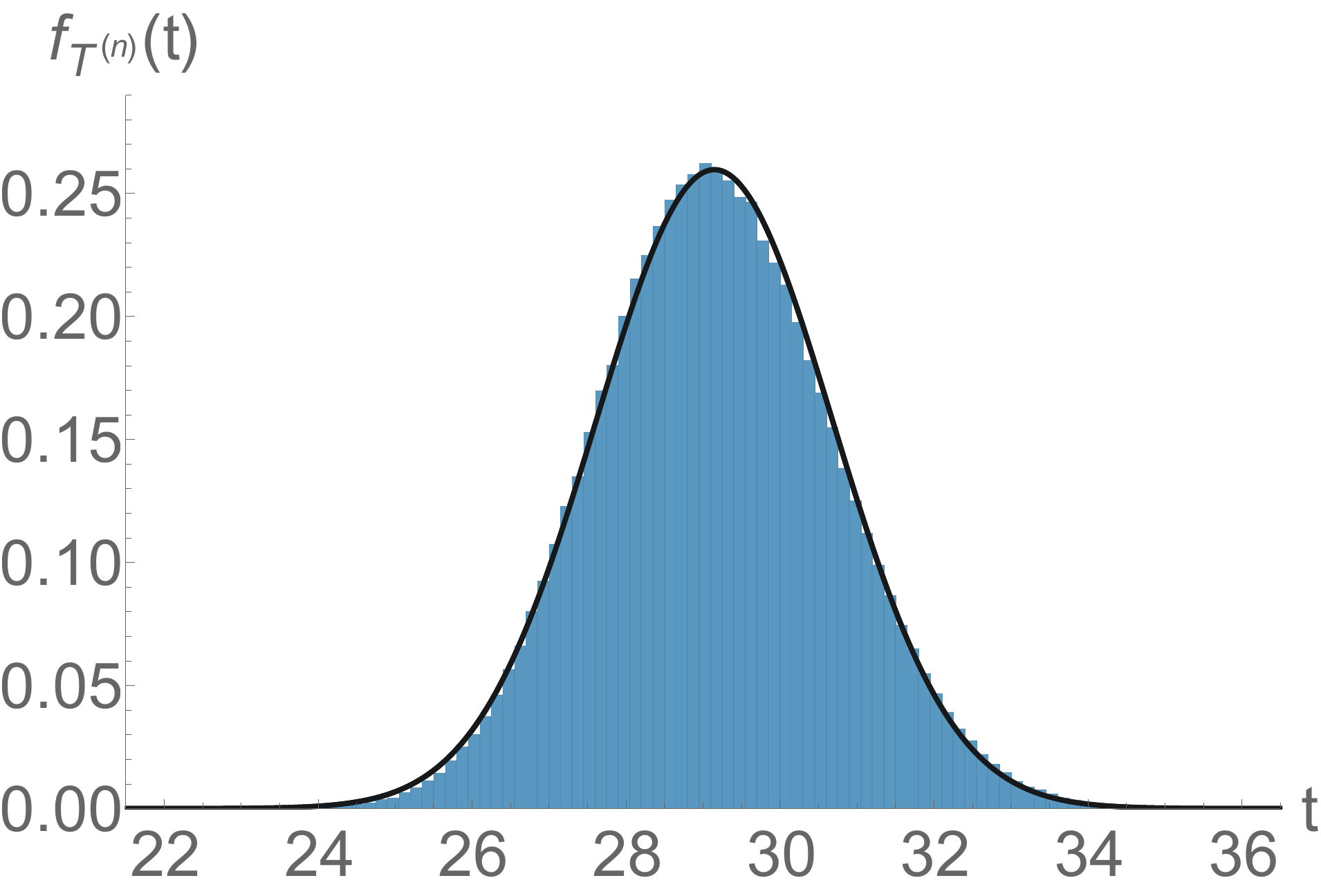}
  \caption*{\begin{small}$
  \eta=\frac{1}{4}$
  \end{small}}
\endminipage\hfill
\minipage{0.32\textwidth}%
  \includegraphics[width=\linewidth]{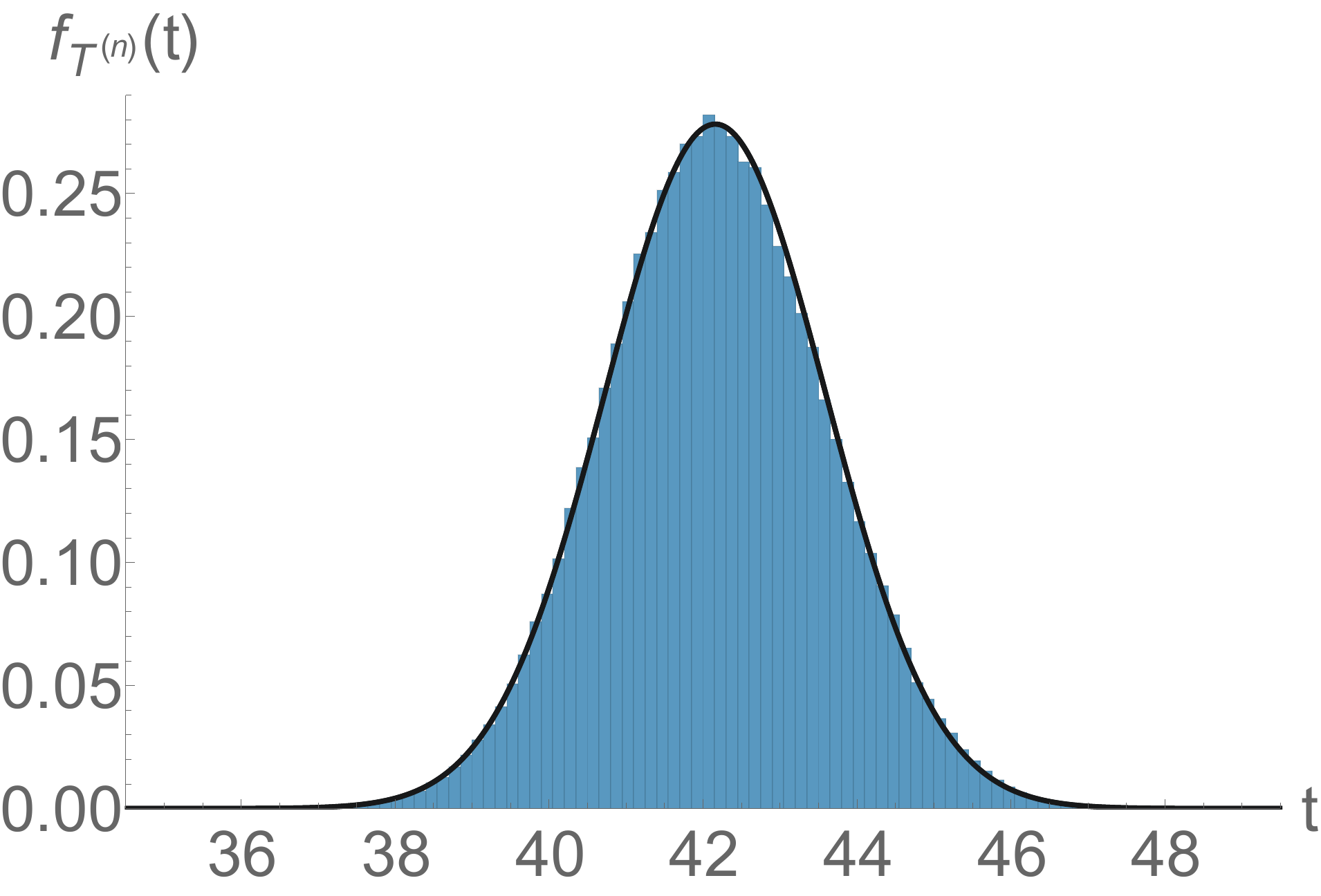}
  \caption*{\begin{small}$
  \eta=\frac{1}{2}$
  \end{small}}
\endminipage
\caption{Density of $T^{(n)}$ for $n=250$ and $R=5$ for different $\eta$: analysis vs. simulation.}\label{fig:T1}
\end{figure}

In Figure \ref{fig:T1} we compare the obtained histogram for the end-to-end delay with the asymptotic result of \eqref{tlim} for the case $R=5$ and $n=250$. We find that the end-to-end delay distribution is approximated well by the normal distribution. Furthermore, setting $\eta=0$ as opposed to $\eta=\frac{1}{2}$ more than halves the expected delay. As predicted by Figure \ref{fig:2} the variance indeed increases as $\eta$ approaches 0. However, the variance does not seem to change significantly. Hence, setting $\eta=0$ seems preferable to $\eta=\frac{1}{2}$.

\begin{figure}[!h]
\minipage{0.32\textwidth}
  \includegraphics[width=\linewidth]{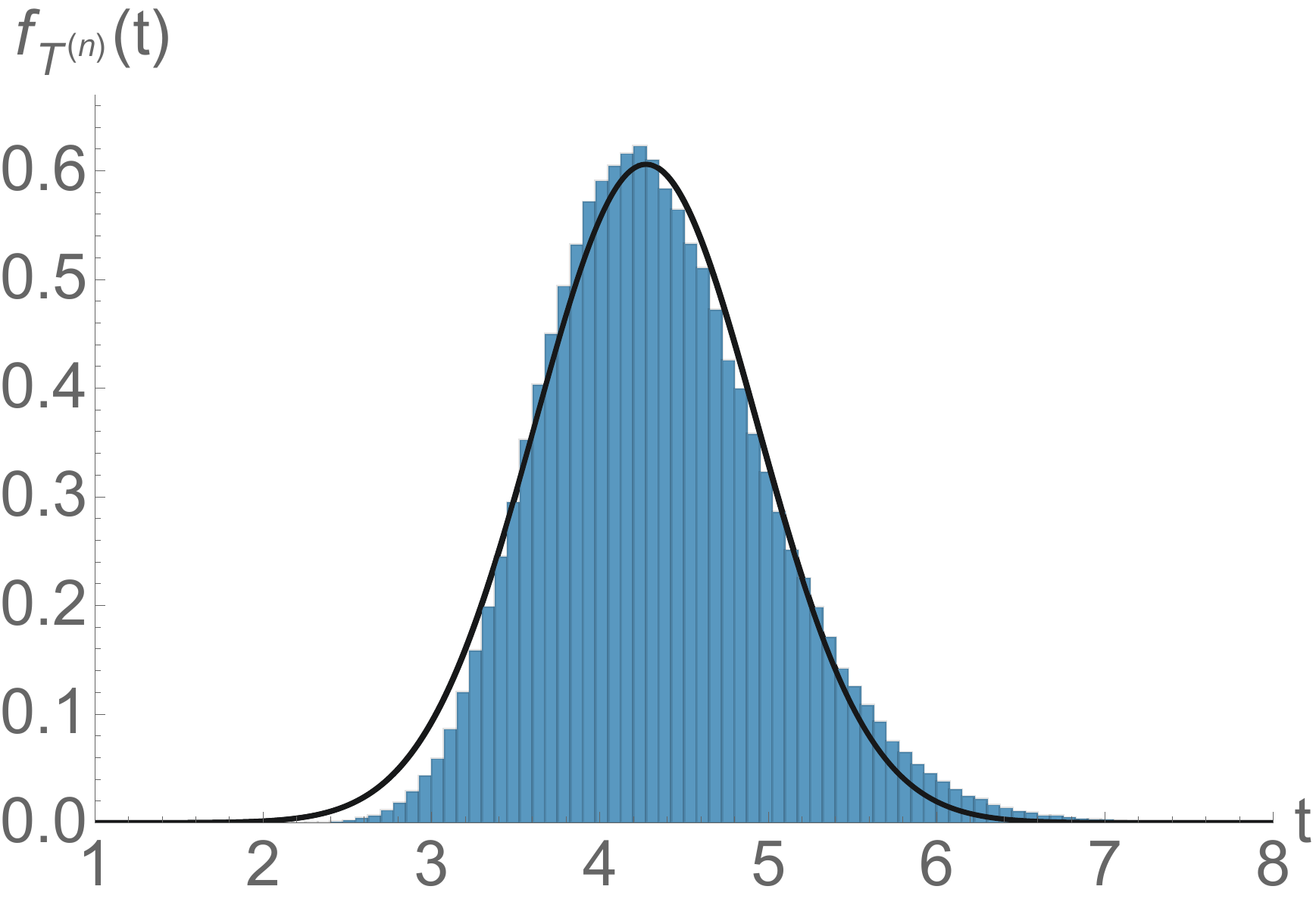}
  \caption*{\begin{small}$
  \eta=0$
  \end{small}}
\endminipage\hfill
\minipage{0.32\textwidth}
  \includegraphics[width=\linewidth]{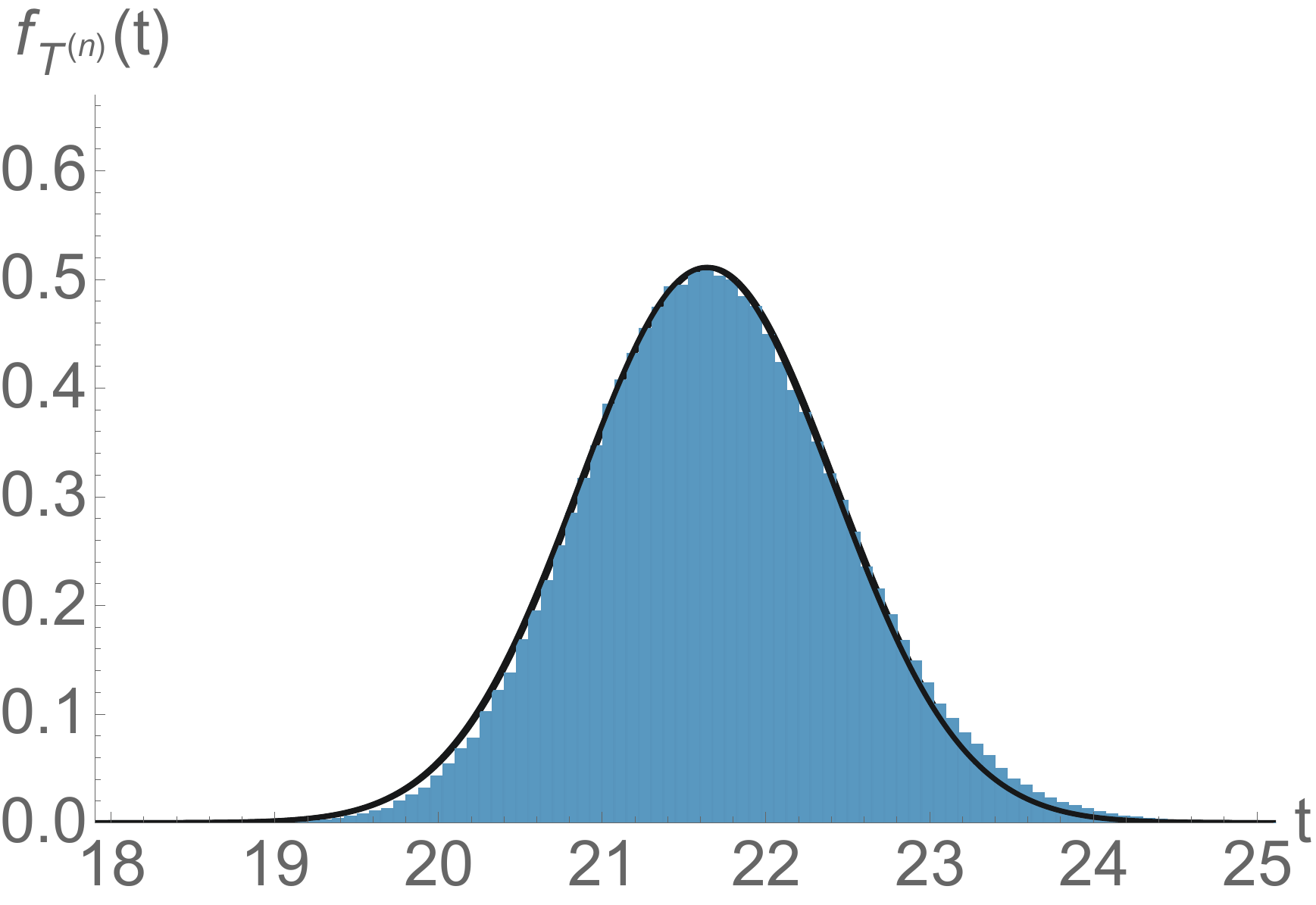}
  \caption*{\begin{small}$
  \eta=\frac{1}{4}$
  \end{small}}
\endminipage\hfill
\minipage{0.32\textwidth}%
  \includegraphics[width=\linewidth]{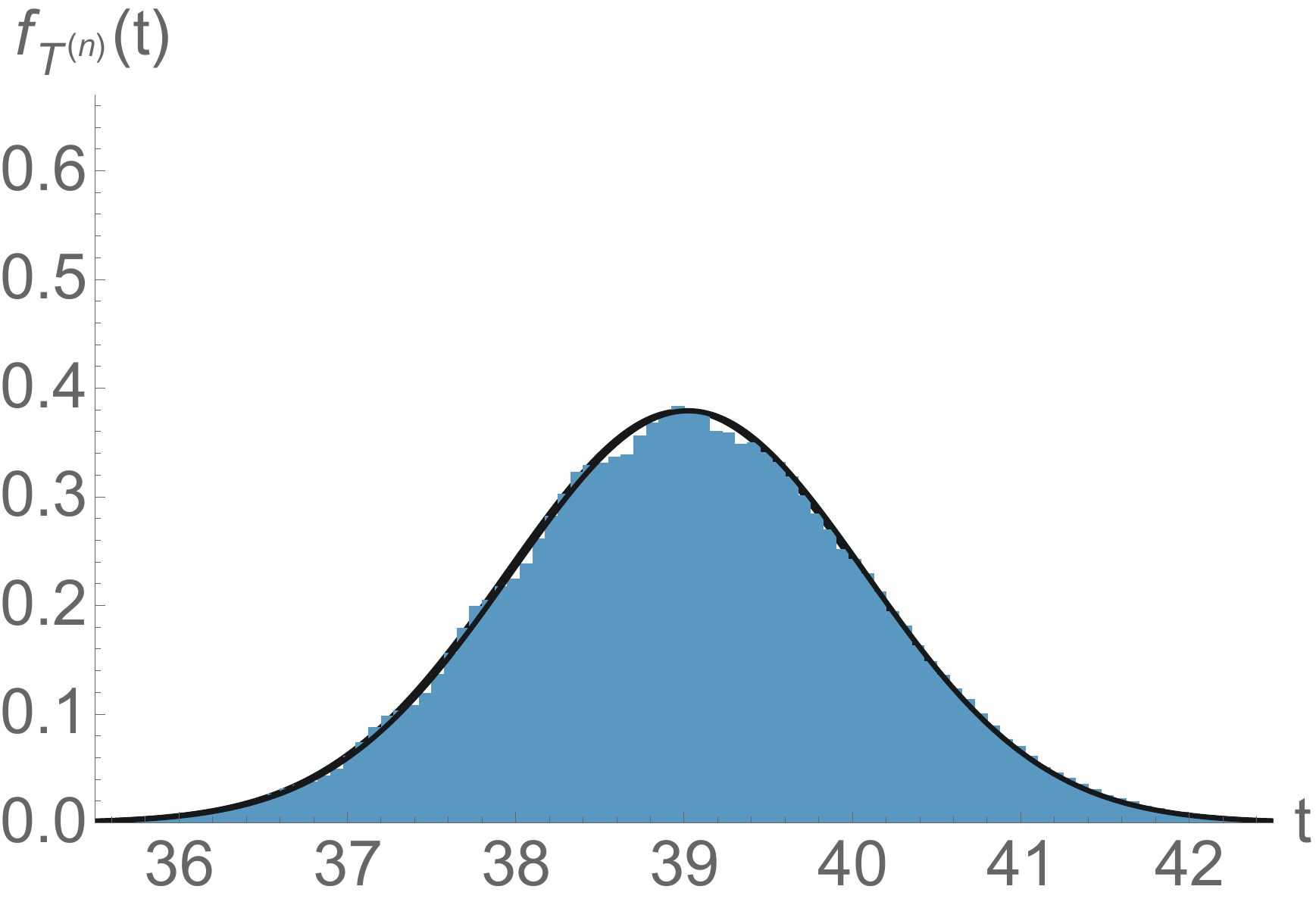}
  \caption*{\begin{small}$
  \eta=\frac{1}{2}$
  \end{small}}
\endminipage
\caption{Density of $T^{(n)}$ for $n=1500$ and $R=30$ for different $\eta$: analysis vs. simulation.}\label{fig:T2}
\end{figure}
  
In Figure \ref{fig:T2} we consider the case $R=30$ and $n=1500$. Also here, we find that the end-to-end delay distribution is approximated well by the normal distribution, although in this case the actual distribution is more skewed. Additionally, we see that in this dense case the Trickle algorithm benefits even more from choosing small $\eta$, giving more than a nine-fold decrease in end-to-end delay when setting $\eta=0$ as opposed to $\eta=1/2$. Finally, as again predicted by Figure \ref{fig:2} the variance indeed decreases as $\eta$ approaches 0. Hence, here setting $\eta=0$ is always preferable to $\eta=\frac{1}{2}$.
  \begin{figure}[!h]\center
\minipage{0.4\textwidth}\center
  \includegraphics[width=0.8\linewidth]{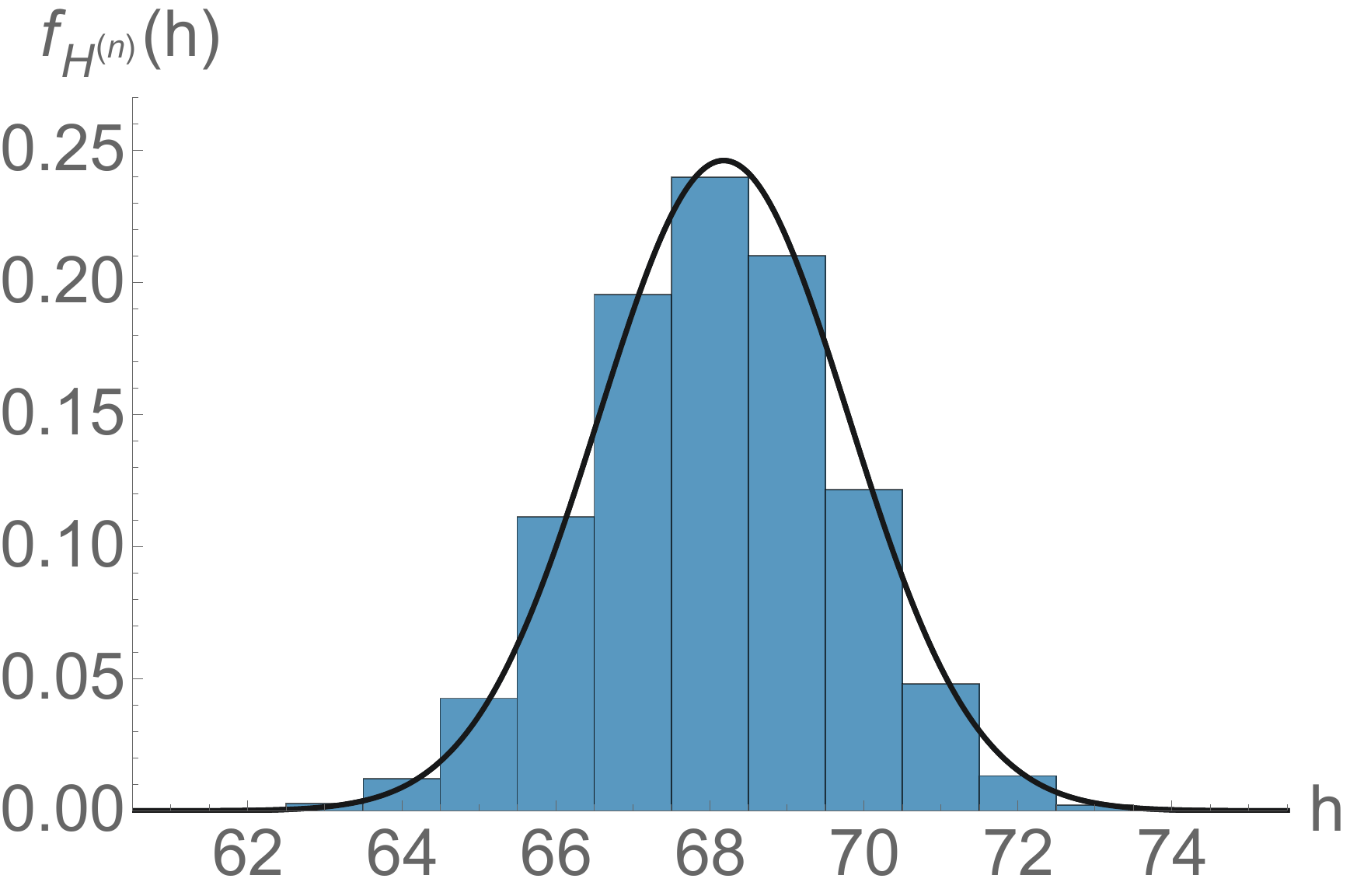}
  \caption*{\begin{small}$
  R=5, n=250$
  \end{small}}
\endminipage
\minipage{0.4\textwidth}\center
  \includegraphics[width=0.8\linewidth]{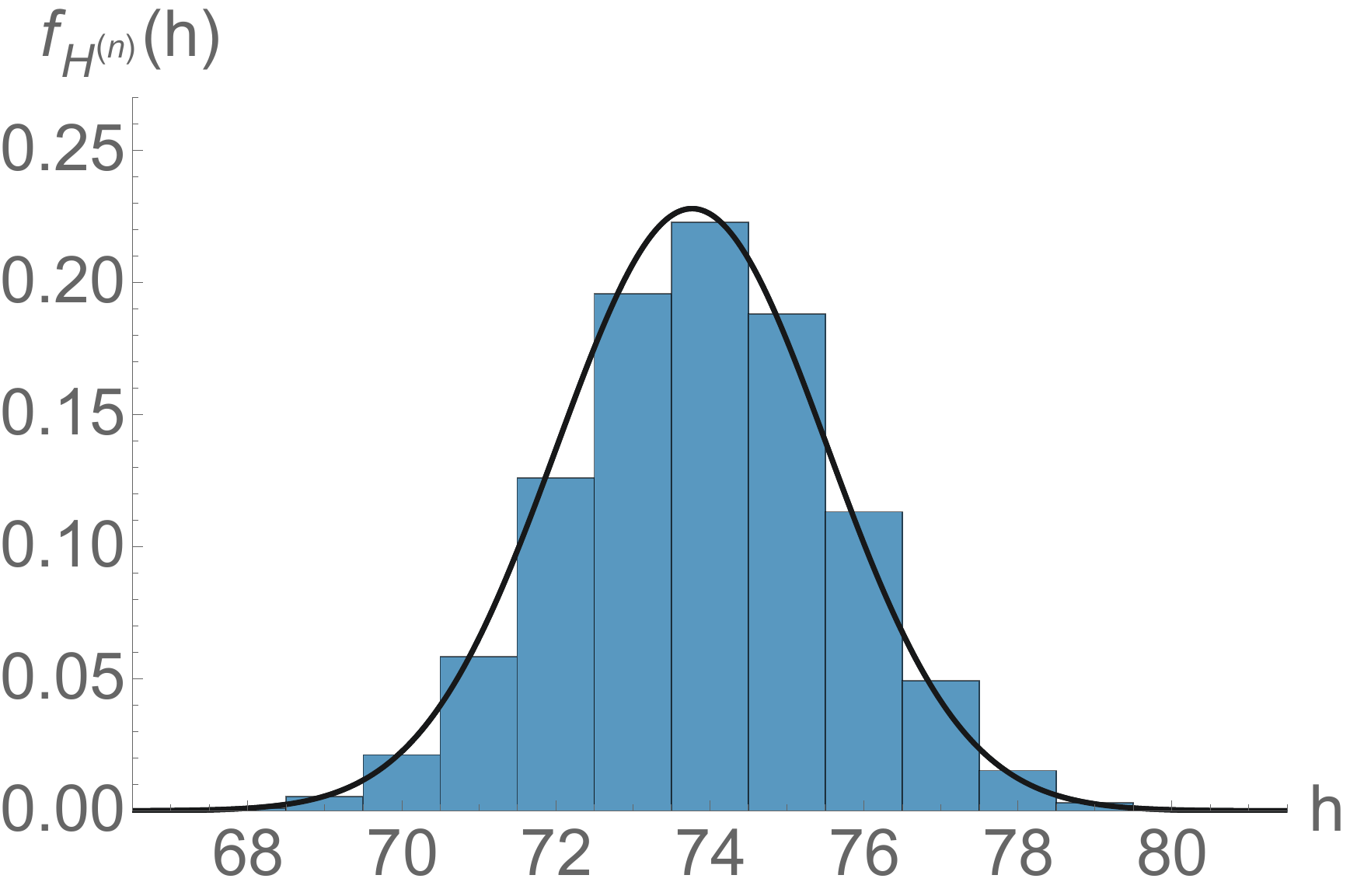}
  \caption*{\begin{small}{$R=30, n=1500$}
  \end{small}}
\endminipage
\caption{Density of $H^{(n)}$ for $R=5$ and $R=30$: analysis vs. simulation.}\label{fig:H}
\end{figure}

Lastly, in Figure \ref{fig:H}, we compare the obtained histograms of the hop-count distributions with the asymptotic result of \eqref{hlim} for both the cases $R=5$ and $R=30$. Also here we find a good match between the simulation and analytical results. Also the increase in variance as $R$ grows is visible.

\section{Hop count and end-to-end delay distribution}\label{gen}
We now focus on deriving the probability generating function of $H^{(n)}$ and the moment generating function of $T^{(n)}$ for finite $n$. These allow us to calculate the exact moments of the hop count and end-to-end delay for small finite $n$, which is not covered by the asymptotic results for large $n$ in the previous section. 

Let us denote by $M_X(s)$ the moment generating function of a continuous variable $X$ and by $G_Y[z]$ the probability generating function of a discrete random variable $Y$. Additionally, for a pair $(X_1,X_2)$ of discrete variables, write
\[G_{(X_1,X_2)}[z_1,z_2]=\sum_{n=0}^\infty\sum_{m=0}^\infty z_1^n z_2^m \mathbb{P}[X_1=n,X_2=m],\]
and for a pair $(X,Y)$, with $X$ a discrete and $Y$ a continuous random variable, write
\[G_{(X,Y)}[z,s]=\int_{t=0}^\infty\sum_{n=0}^\infty z^n e^{st} \text{d}\mathbb{P}[X=n,Y \leq t].\]
We first analyze the probability generating function of $H^{(n)}$.
\subsection{Hop count}
Let $A(m)=\sum_{i=1}^m U_i$, i.e. $A(m)$ is the total number of updated nodes after $m$ transmissions, not including node 0. Furthermore, let $\eta_{ij}$ be the number of transitions in the Markov chain $\{U_m\}_{m=0}^\infty$ between entering state $i$ and entering state $j$ for the first time. Denote by $A(\eta_{ij})$ the number of nodes updated during that time.
\begin{thm}
\[
\mathcal{H}[z_1,z_2]= \sum_{n=0}^\infty \left(\sum_{m=0}^\infty \mathbb{P}\left[H^{(n)}=m\right]z_1^{m}\right)z_2^{n}=\sum_{n=0}^\infty G_{H^{(n)}}[z_1]z_2^{n}
\]
\[=\frac{(z_1-1)z_2}{1-z_2}\left(1+\sum_{j=1}^R\frac{G_{\left(A\left(\eta_{1,j}\right), \eta_{1,j}\right)}[z_1,z_2]}{1-G_{\left(A\left(\eta_{j,j}\right), \eta_{j,j}\right)}[z_1,z_2]}\right)+\frac{1}{1-z_2}.\]
\textbf{Remark}. Note that the probability generating function for $H^{(n)}$ can be obtained by differentiating $\mathcal{H}[z_1,z_2]$:
\begin{equation}\label{MGF_H}G_{H^{(n)}}[z]=\frac{1}{n!}\frac{\text{d}^n}{\text{d}z_2^n}\mathcal{H}[z,z_2]\bigg|_{z_2=0}.\end{equation}
\end{thm}
\begin{proof}
In \cite{renewal5} results are provided which lead to explicit expressions for relevant Laplace transforms of general reward functions for Markov renewal and semi-Markov processes. The following proof closely follows the steps of their proof, but is slightly customized for the case at hand, simplifying the analysis.
 
First, we write
\[\mathbb{P}\left[H^{(n+1)}>m\right]=\sum_{j=1}^R\mathbb{P}\left[A(m)\leq n, U_m=j\mid U_0=1\right].\]
Now, let $\eta_j^{(k)}$ be the time of the $k$'th entrance into state $j$ of the Markov chain $\boldsymbol{U}$, with $\eta_j^{(0)}=0$, that is,
\[\eta_j^{(k)}=\inf\{i>\eta_j^{(k-1)}: U_i=j\}.\]
Then
\[\mathbb{P}[H^{(n+1)}>m]=\mathbbm{1}[m=0]+\sum_{j=1}^R\sum_{k=1}^{\infty}\mathbb{P}\left[A\left(\eta_j^{(k)}\right)\leq n, \eta_j^{(k)}=m \mid U_0=1\right].\]
Substituting and using the fact that $\frac{1}{1-z}\sum_{n=0}^\infty a_n z^n=\sum_{n=0}^\infty \sum_{m=0}^{n} a_m z^n$, we find
\[\sum_{m=0}^\infty \sum_{n=0}^\infty \mathbb{P}[H^{(n+1)}>m]z_1^m z_2^n\]
\[=\frac{1}{1-z_2}+\sum_{m=1}^\infty \sum_{n=0}^\infty\sum_{j=1}^R\sum_{k=1}^{\infty}\mathbb{P}\left[A\left(\eta_j^{(k)}\right)\leq n, \eta_j^{(k)}=m \mid U_0=1\right]z_1^m z_2^n\]
\[=\frac{1}{1-z_2}+\frac{1}{1-z_2}\sum_{m=1}^\infty \sum_{n=0}^\infty\sum_{j=1}^R\sum_{k=1}^{\infty}\mathbb{P}\left[A\left(\eta_j^{(k)}\right)=n, \eta_j^{(k)}=m \mid U_0=1\right]z_1^m z_2^n\]
\[
=\frac{1}{1-z_2}\left(1+\sum_{j=1}^R\sum_{k=1}^{\infty}G_{\left(A\left(\eta_j^{(k)}\right), \eta_j^{(k)}\right)}[z_1,z_2]\right).
\]
Since the consecutive entrance times to a fixed state form a sequence of regeneration times, we obtain for $k\geq1$ and $j=1$
\[
G_{\left(A\left(\eta_1^{(k)}\right), \eta_1^{(k)}\right)}[z_1,z_2]=\left(G_{\left(A\left(\eta_{1,1}\right), \eta_{1,1}\right)}[z_1,z_2]\right)^k,
\]
and for $j\neq 1$ we get
\[ 
G_{\left(A\left(\eta_j^{(k)}\right), \eta_j^{(k)}\right)}[z_1,z_2]=G_{\left(A\left(\eta_{1,j}\right), \eta_{1,j}\right)}[z_1,z_2]\left(G_{\left(A\left(\eta_{j,j}\right), \eta_{j,j}\right)}[z_1,z_2]\right)^{k-1}.
\]
Therefore
\[\sum_{m=0}^\infty \sum_{n=0}^\infty \mathbb{P}[H^{(n)}>m]z_1^m z_2^n
=\frac{z_2}{1-z_2}\left(1+\sum_{j=1}^R\sum_{k=1}^\infty G_{\left(A\left(\eta_{1,j}\right), \eta_{1,j}\right)}[z_1,z_2]\left(G_{\left(A\left(\eta_{j,j}\right), \eta_{j,j}\right)}[z_1,z_2]\right)^{k-1}\right)
\]
\[
=\frac{z_2}{1-z_2}\left(1+\sum_{j=1}^R\frac{G_{\left(A\left(\eta_{1,j}\right), \eta_{1,j}\right)}[z_1,z_2]}{1-G_{\left(A\left(\eta_{j,j}\right), \eta_{j,j}\right)}[z_1,z_2]}\right).\]
Finally we obtain
\[\sum_{m=0}^\infty \sum_{n=0}^\infty \mathbb{P}[H^{(n)}=m]z_1^m z_2^n=\frac{(z_1-1)z_2}{1-z_2}\left(1+\sum_{j=1}^R\frac{G_{\left(A\left(\eta_{1,j}\right), \eta_{1,j}\right)}[z_1,z_2]}{1-G_{\left(A\left(\eta_{j,j}\right), \eta_{j,j}\right)}[z_1,z_2]}\right)+\frac{1}{1-z_2}.\]
\end{proof}

Note that the functions $G_{\left(A\left(\eta_{i,j}\right), \eta_{i,j}\right)}[z_1,z_2]$ can be derived by solving the following system of linear equations:
\begin{equation}\label{G}G_{\left(A\left(\eta_{i,j}\right), \eta_{i,j}\right)}[z_1,z_2]=\sum_{k\neq j}p_{ik}z_1^k z_2 G_{\left(A\left(\eta_{k,j}\right), \eta_{k,j}\right)}[z_1,z_2]+p_{ij}z_1^j z_2,\text{ for all $i$ and $j$}.\end{equation}

\subsection{End-to-end delay}
Similarly, we will now consider the moment generating function for the end-to-end delay $T^{(n)}$. Let $B(t)=\sum_{i=1}^m U_i$ for $t\in[T_{m},T_{m+1})$. Then $B(t)$ is the total number of updated nodes at time $t$, not including node 0. Again let $\eta_{ij}$ be the number of transitions in the Markov chain $\{U_m\}_{m=0}^\infty$ between entering state $i$ and entering state $j$ for the first time and denote by $B(T_{\eta_{ij}})$ the number of nodes updated during that time. Denote by $\nu_j$ the random time the Markov chain stays in state $j$ before transitioning.
\begin{thm}
\begin{equation}\mathcal{T}[z,s]=\sum_{n=0}^\infty \left( \int_{t=0}^{\infty} e^{st}\textnormal{ d}\mathbb{P}[T^{(n)}\leq t]\right)z^{n}
=\sum_{n=0}^\infty M_{T^{(n)}}[s]z^{n}\nonumber
\end{equation}\[=\frac{z}{z-1}\left(1-M_{\nu_1}(s)+\sum_{j=1}^R\left(1-M_{\nu_j}(s)\right)\frac{G_{\left(B\left(T_{\eta_{1,j}}\right), T_{\eta_{1,j}}\right)}[z,s]}{1-G_{\left(B\left(T_{\eta_{j,j}}\right), T_{\eta_{j,j}}\right)}[z,s]}\right)+\frac{1}{1-z}.\]

\textbf{Remark.} Note that the moment generating function for $T^{(n)}$ can be obtained by differentiating $\mathcal{T}[z,s]$:
\begin{equation}\label{MGF_T}M_{T^{(n)}}[s]=\frac{1}{n!}\frac{\text{d}^n}{\text{d}z^n}\mathcal{T}[z,s]\bigg|_{z=0}.
\end{equation}
\end{thm}
\begin{proof}
Again, our proof closely resembles the proof given in \cite{renewal5}. Let $U(t)=U_m$ for $t\in[T_{m},T_{m+1})$. We write
\[\mathbb{P}\left[T^{(n+1)}>t\right]=\sum_{j=1}^R\mathbb{P}\left[B(t)\leq n, U(t)=j\mid U_0=1\right].\]
Again let $\eta_j^{(k)}$ be the time of the $k$'th entry into state $j$ of the Markov chain $\boldsymbol{U}$, with $\eta_j^{(0)}=0$, that is,
\[\eta_j^{(k)}=\inf\{i>\eta_j^{(k-1)}: U_i=j\}.\]
Then
\[\mathbb{P}[T^{(n+1)}>t]=\sum_{j=1}^R\sum_{k=0}^{\infty}\mathbb{P}\left[B\left(T_{\eta_j^{(k)}}\right)\leq n, T_{\eta_j^{(k)}}\leq t<T_{\eta_j^{(k)}+1}, U(t)=j \mid U(0)=1\right].\]
For convenience we write
\[b_j^k(n,t)=\mathbb{P}\left[B\left(T_{\eta_j^{(k)}}\right)\leq n, T_{\eta_j^{(k)}}\leq t<T_{\eta_j^{(k)}+1},U(t)=j\mid U(0)=1\right].\]
Note first that for $k=0$ we have
\[b_j^0(n,t)=\begin{cases}
  \mathbb{P}[\nu_1>t], & \text{if }j=1,\\
  0, & \text{otherwise}.
  \end{cases} \]
  For $k\geq 1$ we can write,
  \[b_j^k(n,t)=\int_{u=0}^\infty\sum_{l=0}^\infty \mathbb{P}[\nu_j>t-u]\mathbbm{1}[n-l\geq 0]\text{d}\mathbb{P}\left[B\left(T_{\eta_j^{(k)}}\right)=l, T_{\eta_j^{(k)}}\leq u\right],\]
  which can be written as a convolution of a function $\phi(u,l)$ and a probability measure:
  \[b_j^k(n,t)=\int_{u=0}^\infty\sum_{l=0}\phi_j(t-u,n-l)\text{d}\mathbb{P}\left[B\left(T_{\eta_j^{(k)}}\right)=l, T_{\eta_j^{(k)}}\leq u\right],\]
  where
  \[\phi_j(u,l)=\mathbb{P}[\nu_j>u]\mathbbm{1}[l\geq 0].\]
  Consequently,
  \[\int_{t=0}^\infty \sum_{n=0}^\infty b_j^k(n,t)z^n e^{s t}\text{d}t=\frac{1-M_{\nu_j}(s)}{s(1-z)}G_{(B(T_{\eta^{(k)}_{j}}), T_{\eta^{(k)}_{j}})}[z,s].\]
Hence,
\[\int_{t=0}^\infty \sum_{n=0}^\infty \mathbb{P}[T^{(n+1)}>t]z^n e^{st}\text{d}t =
 \int_{t=0}^\infty \sum_{n=0}^\infty \mathbb{P}[\nu_1>t]z^n e^{st}\text{d}t+ 
 \int_{t=0}^\infty \sum_{n=0}^\infty\sum_{j=1}^R\sum_{k=1}^{\infty}b_j^k(n,t)z^ne^{st}\text{d}t \]
\[=\frac{1}{s(1-z)}\left(1-M_{\nu_1}(s)+\sum_{j=1}^R\sum_{k=1}^{\infty}\left(1-M_{\nu_j}(s)\right)G_{(B(T_{\eta^{(k)}_{j}}), T_{\eta^{(k)}_{j}})}[z,s]\right).\]
Since the consecutive entry times to a fixed state form a sequence of regeneration times, we obtain for $k\geq 1$ and $j=1$
\[ 
G_{(B(T_{\eta^{(k)}_{1}}), T_{\eta^{(k)}_{1}})}[z,s]=\left(G_{\left(B\left(T_{\eta_{1,1}}\right), T_{\eta_{1,1}}\right)}[z,s]\right)^k,
\]
and for $j\neq 1$ we get
\[
G_{(B(T_{\eta^{(k)}_{j}}), T_{\eta^{(k)}_{j}})}[z,s]=G_{\left(B\left(T_{\eta_{1,j}}\right), T_{\eta_{1,j}}\right)}[z,s]\left(G_{\left(B\left(T_{\eta_{j,j}}\right), T_{\eta_{j,j}}\right)}[z,s]\right)^{k-1}.
\]
Therefore,
\[
\int_{t=0}^\infty \sum_{n=0}^\infty \mathbb{P}[T^{(n)}>t]z^n e^{st}\text{d}t\]
\[=\frac{z}{s(1-z)}\left(1-M_{\nu_1}(s)+\sum_{j=1}^R\sum_{k=1}^{\infty}\left(1-M_{\nu_j}(s)\right)G_{\left(B\left(T_{\eta_{1,j}}\right), T_{\eta_{1,j}}\right)}[z,s]\left(G_{\left(B\left(T_{\eta_{j,j}}\right), T_{\eta_{j,j}}\right)}[z,s]\right)^{k-1}\right)
\]\[=
\frac{z}{s(1-z)}\left(1-M_{\nu_1}(s)+\sum_{j=1}^R\left(1-M_{\nu_j}(s)\right)\frac{G_{\left(B\left(T_{\eta_{1,j}}\right), T_{\eta_{1,j}}\right)}[z,s]}{1-G_{\left(B\left(T_{\eta_{j,j}}\right), T_{\eta_{j,j}}\right)}[z,s]}\right).
\]
Finally we obtain
\[
\sum_{k=0}^\infty \left( \int_{t=0}^{\infty} e^{st}\text{ d}\mathbb{P}[T^{(k)}\leq t]\right)z^{k}\]
\[=\frac{z}{z-1}\left(1-M_{\nu_1}(s)+\sum_{j=1}^R\left(1-M_{\nu_j}(s)\right)\frac{G_{\left(B\left(T_{\eta_{1,j}}\right), T_{\eta_{1,j}}\right)}[z,s]}{1-G_{\left(B\left(T_{\eta_{j,j}}\right), T_{\eta_{j,j}}\right)}[z,s]}\right)+\frac{1}{1-z}.
\]
\end{proof}

Using (\ref{nu}) one can deduce that $M_{\nu_j}(s)$ is the moment generating function of a $\beta(1,j)$ random variable, which can be expressed in terms of the incomplete gamma function $\Gamma[s,x]$ as follows
\begin{equation}
M_{\nu_j}(s)=j! e^s\left(\frac{1}{s(1-\eta)}\right)^{j}\left(1-\frac{\Gamma[j,s(1-\eta)]}{(j-1)!}\right)
\end{equation}
Furthermore, analogous to Equation \eqref{G}, the functions $G_{\left(B\left(T_{\eta_{1,j}}\right), T_{\eta_{1,j}}\right)}[z,s]$ can be derived by solving the following system of equations:
\begin{equation}
G_{\left(B\left(T_{\eta_{i,j}}\right), T_{\eta_{i,j}}\right)}[z,s]=\sum_{k\neq j}p_{ik}z^k M_{\nu_i}[s] G_{\left(B\left(T_{\eta_{k,j}}\right), T_{\eta_{k,j}}\right)}[z,s]+p_{ij}z^j M_{\nu_i}[s],\text{ for all $i$ and $j$}.
\end{equation}

As illustrating examples, in Figure \ref{fig:3} we have plotted the density functions of $H^{(20)}$ and $T^{(20)}$ for $R=4$ and both $\eta=0$ and $\eta=\frac{1}{2}$ obtained by inverting Equations \eqref{MGF_H} and \eqref{MGF_T} using Mathematica. Note that the hop count distribution is the same for both settings.
\begin{figure}[!h]
\minipage{0.33\textwidth}
  \includegraphics[width=\linewidth]{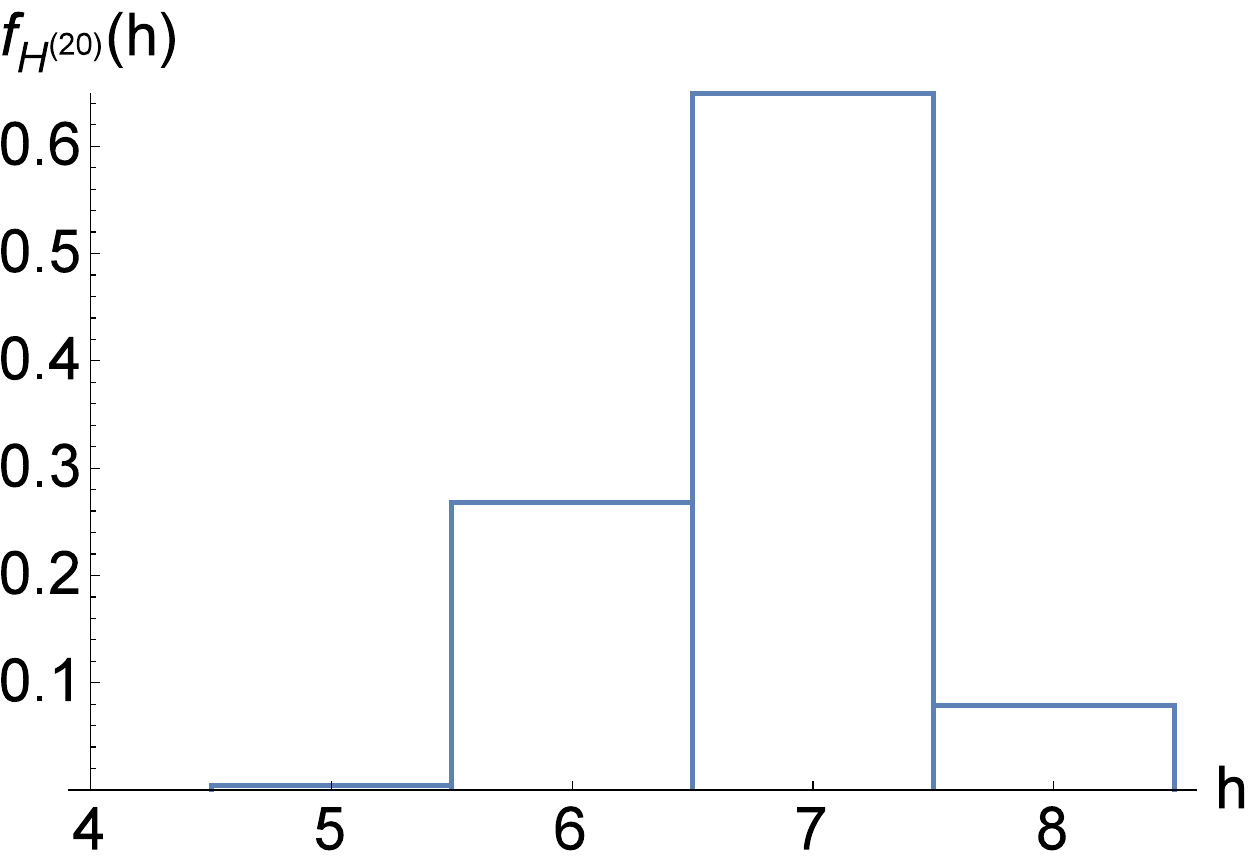}
  \caption*{\begin{small}Hop count
  \end{small}}
\endminipage%
\minipage{0.33\textwidth}%
  \includegraphics[width=\linewidth]{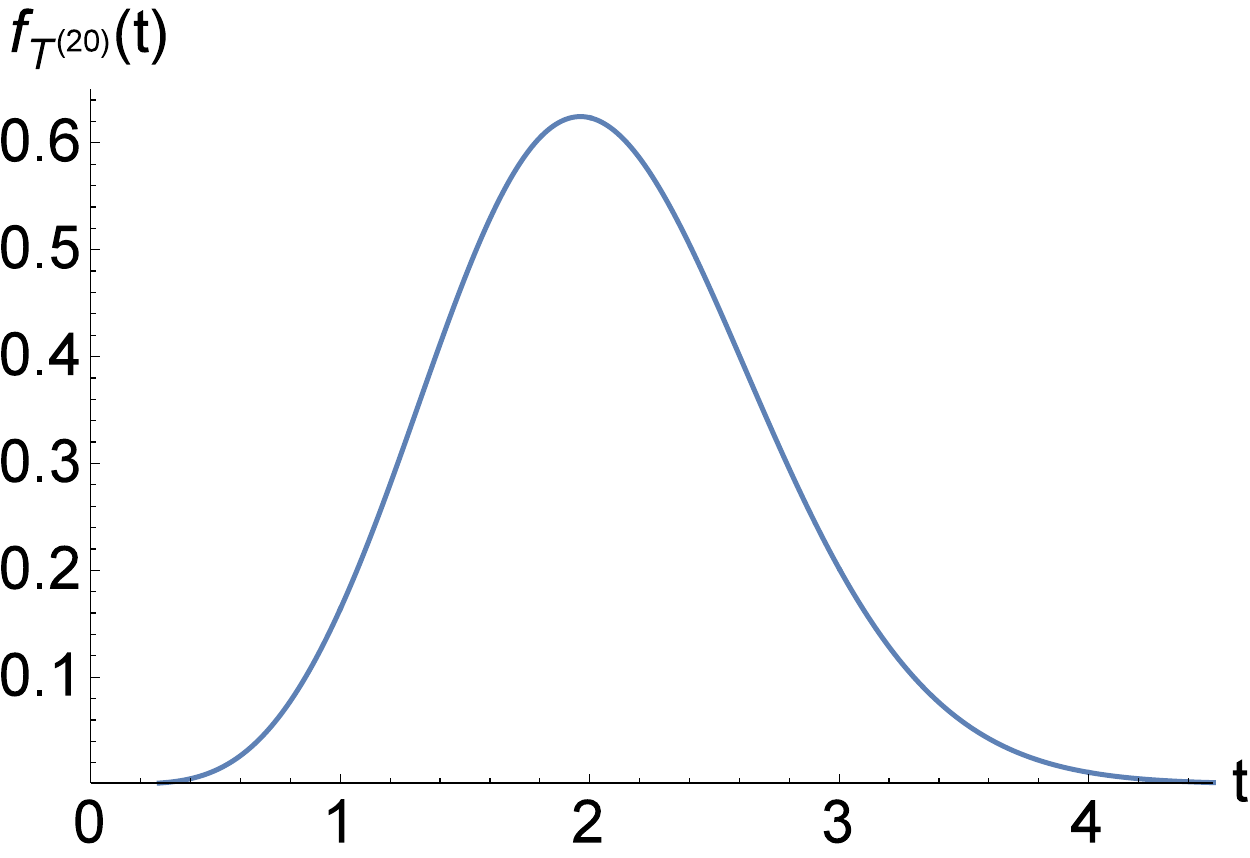}
  \caption*{\begin{small}End-to-end delay ($\eta=0$)
  \end{small}}
\endminipage
\minipage{0.33\textwidth}%
  \includegraphics[width=\linewidth]{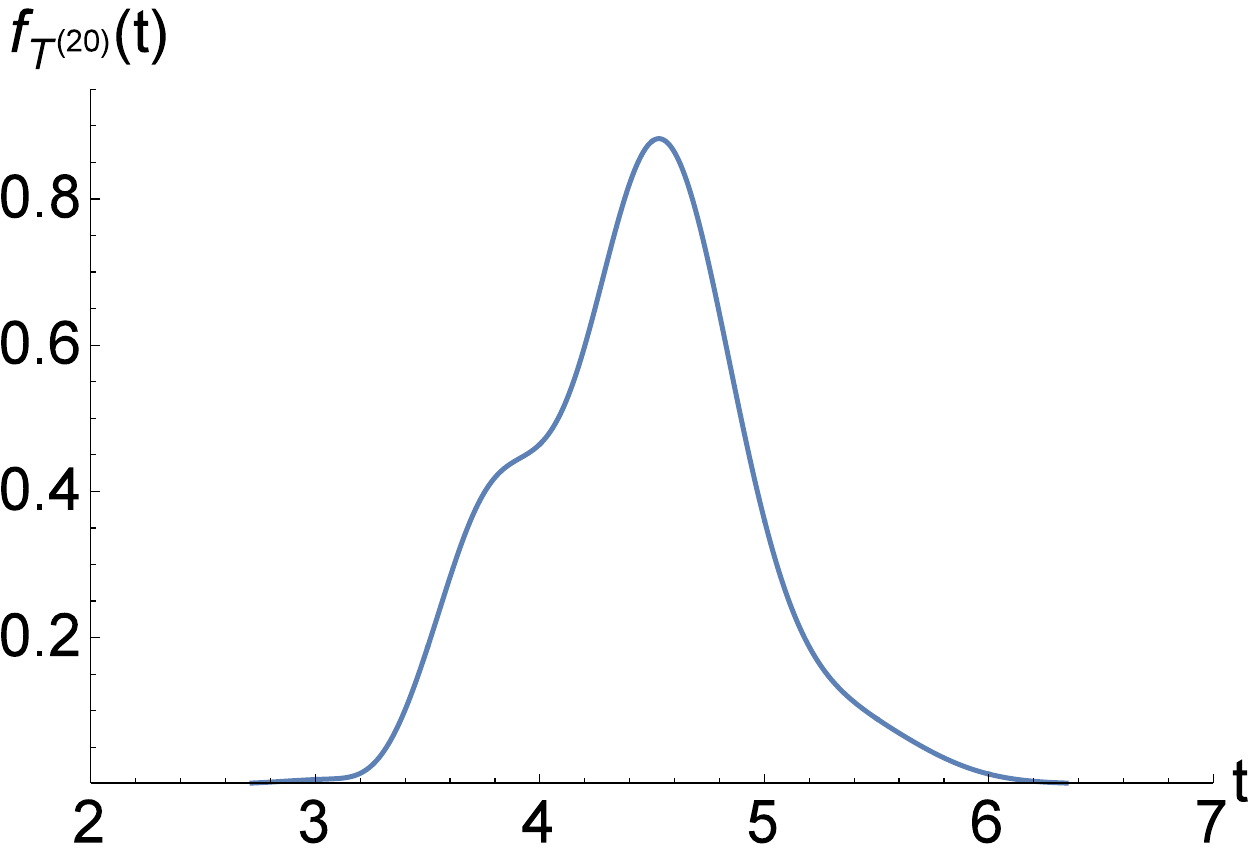}
  \caption*{\begin{small}End-to-end delay ($\eta=\frac{1}{2}$)
  \end{small}}
\endminipage
\caption{Hop count and end-to-end delay density for $n=20$ and $R=4$.}\label{fig:3}
\end{figure}
\section{Conclusion}\label{conclusion}
In this paper, we presented a generalized version of the Trickle algorithm with an additional parameter $\eta$, which allows us to set the length of a listen-only period for newly updated nodes. We argue that this parameter can greatly increase the speed at which the Trickle algorithm can disseminate data, while retaining scalability. These claims are supported by a mathematical analysis and simulations of a Trickle propagation event in line networks.

First, we provided an analysis of the hop count and end-to-end delay distribution for line networks consisting of $n$ nodes. We derived formulas for the mean and variance of the hop count and end-to-end delay as a function of $R$, $n$ and $\eta$, giving insight into the performance of the Trickle algorithm. Additionally, we showed that both distributions converge to a normal distribution as $n$ goes to infinity.

Secondly, we demonstrated how to derive explicit expressions for the probability and moment generating functions of the hop count and end-to-end delay. As was shown, these functions can be used to determine the respective density functions for small network sizes explicitly. 

From our analysis we can conclude that the generalized version of Trickle as presented in this paper with $\eta=0$ allows for better performance in terms of end-to-end delay, compared to the original description of Trickle. It greatly decreases end-to-end delay, while having only a small effect on its variability and the energy consumption of the network.

Finally, we note that our analysis only provides a first step towards a complete understanding of Trickle's propagation performance, since the analysis in this work is restricted to line networks. The impact of network topology and MAC-layer interactions on the performance of the Trickle algorithm remain as interesting topics for further research.

\appendix
\section{Proof of Theorem \ref{thm1}}\label{appA}
A result for first passage times of Markov renewal processes implies (see \cite{renewal2}, Theorem 3.4)
\[\text{Var}\left[H^{(n)}\right]\sim n\gamma_U^2/\mu_U^3\text{, as $n\rightarrow 
\infty$}.\]
Therefore we need to show that $\gamma_U^2=\frac{1}{54}(R^2+R-2)$. To simplify the analysis we will assume stationarity of the Markov chain $\{U_i\}_{i=0}^\infty$, but again note that the final result will also hold for the non-stationary case $U_0=1$. We show by induction that $\text{Cov}[U_0,U_{j}]=\left(-\frac{1}{2}\right)^j\frac{1}{18}(R^2+R-2)$. First note that 
\[\text{Cov}[U_0,U_{j}]=\mathbb{E}[U_0 U_{j}]-\mu_U^2=\mathbb{E}[U_0 U_{j}]-\frac{1}{9}(2R+1)^2.\]
For $j=0$ we find
\[\mathbb{E}[U_0^2]=\sum_{k=1}^R\pi_k k^2=\frac{2}{R(R+1)}\sum_{k=1}^Rk^3=\frac{1}{2}R(1+R).\]
Hence $\text{Cov}[U_0,U_0]=\frac{1}{2}R(1+R)-\frac{1}{9}(2R+1)^2=\frac{1}{18}(R^2+R-2)$, which is our induction basis. Now let $p_{ij}^{(k)}$ be the probability that starting from state $i$ the Markov chain is in state $j$ after $k$ steps. Then we can write
\[\mathbb{E}[U_0U_{j}]=\sum_{k=1}^R \pi_k \left( \sum_{l=1}^R p_{kl}^{(j)} kl \right)=\sum_{k=1}^R \pi_k \left( \sum_{m=1}^R p_{km}^{(j-1)}\left(\sum_{l=1}^R p_{ml} kl \right)\right)\]
\[=\sum_{k=1}^R \pi_k \left( \sum_{m=1}^R p_{km}^{(j-1)}\left(\sum_{l=R-m+1}^R \frac{1}{m} kl \right)\right)=\sum_{k=1}^R \pi_k \left( \sum_{m=1}^R p_{km}^{(j-1)}\frac{1}{2}(2R+1-m)k \right)\]
\[=\frac{1}{2}(2R+1)\sum_{k=1}^R \pi_k k-\frac{1}{2}\sum_{k=1}^R \pi_k \left( \sum_{m=1}^R p_{km}^{(j-1)}mk \right)=\frac{1}{6}(2R+1)^2-\frac{1}{2}\mathbb{E}[U_0U_{j-1}].\]
Consequently, we find
\[\text{Cov}[U_0,U_{j}]=-\frac{1}{2}\mathbb{E}[U_0U_{j-1}]+\frac{1}{18}(2R+1)^2=\left(-\frac{1}{2}\right)\text{Cov}[U_0,U_{j-1}].\]
Using this result it is easy to see that
\[\gamma_U^2=\frac{1}{18}(R^2+R-2)+2\sum_{j=1}^\infty\left(-\frac{1}{2}\right)^j\frac{1}{18}(R^2+R-2)=\frac{1}{54}(R^2+R-2),\] 
which completes the proof.

\section{Proof of Theorem \ref{thm2}}\label{appB}
A result for stopped functionals of Markov renewal processes implies (see \cite{renewal3}, Theorem 2)
\[\text{Var}\left[T^{(n)}\right]\sim n\gamma_T^2/\mu_U^3\text{, as $n\rightarrow 
\infty$}.\]
Now, rewriting \eqref{gammaT} we have
\[
\gamma_T^2=\text{Var}[U_0\mu_\theta-\theta_1\mu_U]+2\sum_{i=1}^\infty\text{Cov}[U_0\mu_\theta-\theta_1\mu_U,U_i\mu_\theta-\theta_{i+1}\mu_U]=\mu_\theta^2\gamma_U^2+\mu_U^2\gamma_\theta^2-2\mu_U\mu_\theta\Delta.\]
Here $\gamma_\theta$ and $\Delta$ are defined as
\[
\gamma_\theta^2=\lim_{m\rightarrow \infty}\frac{1}{m}\textnormal{Var}\left[T_{m+1}\right]=\text{Var}[\theta_1]+2 \sum_{j=1}^\infty \textnormal{Cov}[\theta_1,\theta_{j+1}],
\]
\[\Delta=\text{Cov}[\theta_1,U_0]
+\sum_{j=1}^\infty \textnormal{Cov}[U_0,\theta_{j+1}]+\sum_{j=1}^\infty \textnormal{Cov}[\theta_1,U_{j}].\]
For simplicity of the analysis, assume again stationarity of the Markov chain $\{U_i\}_{i=0}^\infty$. An expression for $\gamma_\theta^2$ in terms of the matrix $M$ and the fundamental matrix $Z=(I-P+\boldsymbol{1}\boldsymbol{\pi})^{-1}$ is given in \cite{renewal4}, that is
\[\gamma_\theta^2=\text{Var}[\theta_1]+2\boldsymbol{\pi}MZM\boldsymbol{1}-2\mu_\theta^2.\]
Here,
\[\text{Var}[\theta_1]=4(1-\eta)^2\left(\frac{6+R}{8+4R}-\left(\frac{2+R}{2R}-\frac{\sum_{j=1}^{R+1}\frac{1}{j}}{R(1+R)}\right)^2\right).\]
Finally, analogous to the proof of Theorem \ref{thm1}, we can show that $\textnormal{Cov}[\theta_1,U_{j}]=\left(-\frac{1}{2}\right)^j\textnormal{Cov}[\theta_1,U_0]$ and
\[\text{Cov}[\theta_1,U_0]=(1-\eta)\frac{(4R+8)\left(\sum_{j=1}^{R+1}\frac{1}{j}\right)-(R^2+9R+8)}{3R^2+3R}.\]
Now, since  $\pi_i p_{ij}=\pi_jp_{ji}$ for all $i$ and $j$, the Markov chain $\mathbf{U}$ is reversible (see \cite{kelly}, Theorem 1.2). This implies $(U_0,\theta_{j+1})\sim(\theta_1,U_{j})$ and hence we have $\textnormal{Cov}[U_0,\theta_{j+1}]=\textnormal{Cov}[\theta_1,U_{j}]$. This then yields
\[\Delta=\text{Cov}[\theta_1,U_0]
+\sum_{j=1}^\infty \textnormal{Cov}[U_0,\theta_{j+1}]+\sum_{j=1}^\infty \textnormal{Cov}[\theta_1,U_{j}]
=\text{Cov}[\theta_1,U_0]
+2\sum_{j=1}^\infty \textnormal{Cov}[\theta_1,U_{j}]\]\[=(1-\eta)\frac{(4R+8)\left(\sum_{j=1}^{R+1}\frac{1}{j}\right)-(R^2+9R+8)}{9R^2+9R}.
\]

\bibliographystyle{abbrv}
\bibliography{refs}

\begin{thebibliography}{10}

\bibitem{trick1}
I.~F. Akyildiz, W.~Su, Y.~Sankarasubramaniam, and E.~Cayirci.
\newblock Wireless sensor networks: a survey.
\newblock {\em Computer Networks}, 38(4):393--422, March 2002.

\bibitem{renewal3}
G.~Alsmeyer and A.~Gut.
\newblock Limit theorems for stopped functionals of {M}arkov renewal processes.
\newblock {\em Annals of the Institute of Statistical Mathematics},
  51(2):369--382, 1999.

\bibitem{trick7}
M.~Becker, K.~Kuladinithi, and C.~G{\"o}rg.
\newblock Modelling and simulating the {T}rickle algorithm.
\newblock In {\em Proceedings of MONAMI}, pages 135--144. Springer, 2011.

\bibitem{trick13}
Y.~Busnel, M.~Bertier, E.~Fleury, and A.-M. Kermarrec.
\newblock {GCP}: gossip-based code propagation for large-scale mobile wireless
  sensor networks.
\newblock In {\em Proceedings of the 1st International Conference on Autonomic
  Computing and Communication Systems}, Autonomics '07, pages 11:1--11:5, 2007.

\bibitem{renewal}
E.~{\c C}inlar.
\newblock Markov renewal theory.
\newblock {\em Advances in Applied Probability}, 1(2):123--187, 1969.

\bibitem{sims}
T.~Clausen, A.~C. de~Verdiere, and J.~Yi.
\newblock Performance analysis of {T}rickle as a flooding mechanism.
\newblock In {\em Proceedings of the 15th IEEE International Conference on
  Communication Technology (ICCT)}, pages 565--572, Nov 2013.

\bibitem{trick17}
W.~Dong, Y.~Liu, X.~Wu, L.~Gu, and C.~Chen.
\newblock Elon: enabling efficient and long-term reprogramming for wireless
  sensor networks.
\newblock In {\em Proceedings of the ACM SIGMETRICS International Conference on
  Measurement and Modeling of Computer Systems}, SIGMETRICS '10, pages 49--60,
  2010.

\bibitem{renewal2}
C.-D. Fuh and T.~L. Lai.
\newblock Asymptotic expansions in multidimensional {M}arkov renewal theory and
  first passage times for {M}arkov random walks.
\newblock {\em Advances in Applied Probability}, 33(3):652--673, 2001.

\bibitem{trick9}
J.~Hui and R.~Kelsey.
\newblock {M}ulticast {P}rotocol for {L}ow power and {L}ossy {N}etworks
  ({MPL}).
\newblock {\em {IETF}, {I}nternet-{D}raft
  draft-ietf-roll-trickle-mcast-07.txt}, February 2014.

\bibitem{trick15}
J.~W. Hui and D.~Culler.
\newblock The dynamic behavior of a data dissemination protocol for network
  programming at scale.
\newblock In {\em Proceedings of the 2nd International Conference on Embedded
  Networked Sensor Systems}, SenSys '04, pages 81--94, 2004.

\bibitem{renewal4}
J.~Keilson and D.~G.~M. Wishart.
\newblock Addenda to processes defined on a finite {M}arkov chain.
\newblock {\em Mathematical Proceedings of the Cambridge Philosophical
  Society}, 63:187--193, 1967.

\bibitem{kelly}
F.~P. Kelly.
\newblock {\em Reversibility and Stochastic Networks}.
\newblock Cambridge University Press, New York, 2011.

\bibitem{DIS}
H.~Kermajani and C.~Gomez.
\newblock On the network convergence process in {RPL} over {IEEE} 802.15.4
  multihop networks: Improvement and trade-offs.
\newblock {\em Sensors}, 14(7):11993--12022, 2014.

\bibitem{trick6}
H.~Kermajani, C.~Gomez, and M.~H. Arshad.
\newblock Modeling the message count of the {T}rickle algorithm in a
  steady-state, static wireless sensor network.
\newblock {\em Communications Letters, IEEE}, 16(12):1960--1963, December 2012.

\bibitem{trick5}
P.~Levis, E.~Brewer, D.~Culler, D.~Gay, S.~Madden, N.~Patel, J.~Polastre,
  S.~Shenker, R.~Szewczyk, and A.~Woo.
\newblock The emergence of a networking primitive in wireless sensor networks.
\newblock {\em Communications of the ACM}, 51(7):99--106, July 2008.

\bibitem{trick4}
P.~Levis, T.~Clausen, J.~Hui, O.~Gnawali, and J.~Ko.
\newblock The {T}rickle algorithm.
\newblock {\em Internet RFC 6206}, March 2011.

\bibitem{trick3}
P.~Levis, N.~Patel, D.~Culler, and S.~Shenker.
\newblock Trickle: A self-regulating algorithm for code propagation and
  maintenance in wireless sensor networks.
\newblock In {\em Proceedings of the First USENIX/ACM Symposium on Networked
  Systems Design and Implementation}, pages 15--28, 2004.

\bibitem{thesis}
T.~M.~M. Meyfroyt.
\newblock Modeling and analyzing the {T}rickle algorithm. {M}aster's thesis,
  {E}indhoven {U}niversity of {T}echnology, {E}indhoven, {T}he {N}etherlands.
\newblock
  \url{http://alexandria.tue.nl/extra1/afstversl/wsk-i/meyfroyt2013.pdf}, August
  2013.

\bibitem{meyfroyt}
T.~M.~M. Meyfroyt, S.~C. Borst, O.~J. Boxma, and D.~Denteneer.
\newblock Data dissemination performance in large-scale sensor networks.
\newblock {\em SIGMETRICS Performance Evaluation Review}, 42(1):395--406, June
  2014.

\bibitem{renewal5}
V.~T. Stefanov.
\newblock Exact distributions for reward functions on semi-{M}arkov and
  {M}arkov additive processes.
\newblock {\em Journal of Applied Probability}, 43(4):1053--1065, 2006.

\bibitem{trickf}
C.~Vallati and E.~Mingozzi.
\newblock {Trickle-F: Fair broadcast suppression to improve energy-efficient route formation with the RPL routing protocol}.
\newblock In {\em Proceedings of the 3rd IFIP Conference on Sustainable Internet and ICT for Sustainability (SustainIT)}, pages 1--9, October 2013.

\bibitem{trick11}
C.-Y. Wan, A.~T. Campbell, and L.~Krishnamurthy.
\newblock Pump-slowly, fetch-quickly {(PSFQ)}: A reliable transport protocol
  for sensor networks.
\newblock {\em IEEE Journal on Selected Areas in Communications},
  23(4):862--872, April 2005.

\bibitem{trick2}
T.~Winter, P.~Thubert, A.~Brandt, J.~Hui, R.~Kelsey, P.~Levis, K.~Pister,
  R.~Struik, J.~Vasseur, and R.~Alexander.
\newblock {RPL}: {IPv6} routing protocol for low-power and lossy networks.
\newblock {\em Internet RFC 6550}, March 2012.

\bibitem{trick12}
Y.~Yu, L.~J. Rittle, V.~Bhandari, and J.~B. Lebrun.
\newblock Supporting concurrent applications in wireless sensor networks.
\newblock In {\em Proceedings of the 4th International Conference on Embedded
  Networked Sensor Systems}, pages 139--152, 2006.


\end{thebibliography}
\end{document}